\documentclass[hidelinks,11pt,a4paper]{article}

\usepackage[top=3cm, bottom=3cm, left=2.3cm, right=2.3cm,
footskip=.7in]{geometry}

\usepackage{amsmath}
\usepackage{amssymb}
\usepackage{amsfonts,amsthm}
\usepackage{thmtools,thm-restate}
\usepackage{tabularx,lipsum,environ}
\usepackage{multirow}
\usepackage{hyperref}
\usepackage{graphicx,float}
\usepackage{enumitem,linegoal}
\usepackage{caption}
\usepackage{subcaption}
\usepackage{todonotes}
\usepackage[T1]{fontenc}

\usepackage{lmodern}
\usepackage{setspace}

\newcommand{\NP}{{\sf NP}}

\newcommand{\diam}{{\rm diam}}

\newcommand{\rvc}{{\mathbf{rvc}}}
\newcommand{\srvc}{{\mathbf{srvc}}}

\newcommand{\krvc}{{\sc{RVC}}}
\newcommand{\ksrvc}{{\sc{SRVC}}}

\newcommand{\calE}{{\mathcal{E}}}

\newcommand{\eL}{\mathcal{L}}

\usepackage{pgfplots}
\usetikzlibrary{calc,intersections,arrows}

\tikzstyle{path}=[rectangle, fill=gray, inner sep=2pt, minimum width=4pt]
\tikzstyle{edge} = [line width = 1pt]

\tikzstyle{line} = [line width = 1.3pt]
\tikzstyle{segment} = [] 
\tikzstyle{thickline} = [line width = 2pt]
\tikzstyle{possiblesegment} = [dashed] 
\tikzstyle{vertex}=[circle, draw, fill=black, inner sep=0pt, minimum width=4pt]

\newcommand{\ellips}[2]{
	\def\a{.5} \def\b{1}    
	\def\cx{#1} \def\cy{#2}  
	\def\xp{-2} \def\yp{0} 
	
	\coordinate (P) at ({\xp+\cx},{\yp+\cy});
	\coordinate (A) at ({\a+\cx},\cy);   
	\coordinate (B) at ({-\a+\cx},\cy);  
	\coordinate (C) at (\cx,{-\b+\cy});   

	\draw[name path=ellipse,fill=gray!30,opacity=0.5](\cx,\cy)circle[x radius=\a,y radius=\b];
	
	\path[name path=linePC] (P)--(C);
	
	\path [name intersections={of = ellipse and linePC}];
	\coordinate (E) at (intersection-1);
	\coordinate (F) at (intersection of A--C and B--E);
	\path let \p1=(F) in node (G) at (\x1,1.2*\b){};
	
	\path[name path=lineFG] (F)--(G);
	\path [name intersections={of = ellipse and lineFG}];
	\coordinate (X) at (intersection-1);
	\coordinate (Y) at (intersection-2);
	\draw (X)--(P);
	\draw (Y) -- (P);
}
\newcommand{\ellipshigh}[2]{
	\def\a{.5} \def\b{1}    
	\def\cx{#1} \def\cy{#2}  
	\def\xp{-2} \def\yp{0.8} 
	
	\coordinate (P) at ({\xp+\cx},{\yp+\cy});
	\coordinate (A) at ({\a+\cx},\cy);   
	\coordinate (B) at ({-\a+\cx},\cy);  
	\coordinate (C) at (\cx,{-\b+\cy});   

	\draw[name path=ellipse,fill=gray!30,opacity=0.5](\cx,\cy)circle[x radius=\a,y radius=\b];
	
	\path[name path=linePC] (P)--(C);
	
	\path [name intersections={of = ellipse and linePC}];
	\coordinate (E) at (intersection-1);
	\coordinate (F) at (intersection of A--C and B--E);
	\path let \p1=(F) in node (G) at (\x1,1.2*\b){};
	
	\path[name path=lineFG] (F)--(G);
	\path [name intersections={of = ellipse and lineFG}];
	\coordinate (X) at (intersection-1);
	\coordinate (Y) at (intersection-2);
	\draw (Y) -- (P);
	\draw (P) -- (\cx,1);
}

\usepackage{enumitem}

\usepackage{boxedminipage,tikz}

\setlist[itemize]{noitemsep}
\setlist[enumerate]{noitemsep}

\usepackage{authblk}

\newtheorem{theorem}{Theorem}[section]
\newtheorem{theoremintro}{Theorem}

\newtheorem{lemma}[theorem]{Lemma}

\newtheorem{corollary}[theorem]{Corollary}

\newtheorem{claim}{Claim}

\newtheorem{conjecture}[theorem]{Conjecture}

\newtheorem{observation}{Observation}

\theoremstyle{definition}

\newenvironment{claimproof}{\begin{proof}\renewcommand{\qedsymbol}{\claimqed}}{\end{proof}\renewcommand{\qedsymbol}{\plainqed}}
\let\plainqed\qedsymbol

\title{\textbf{\Large Algorithms for the rainbow vertex coloring problem on graph classes}}

\author[1]{Paloma T.\ Lima}
\author[2]{Erik~Jan~van~Leeuwen}
\author[2,3]{Marieke van der Wegen}

\affil[1]{Department of Informatics, University of Bergen, Norway}

\affil[2]{Department of Information and Computing Sciences, Utrecht University, The Netherlands}

\affil[3]{Mathematical Institute, Utrecht University, The Netherlands}

\affil[ ]{\texttt{paloma.lima@uib.no, \{e.j.vanleeuwen,m.vanderwegen\}@uu.nl}}

\date{}

\begin{document}
	\setstretch{1.05}
	\maketitle
	
	\begin{abstract}
		Given a vertex-colored graph, we say a path is a rainbow vertex path if all its internal vertices have distinct colors. The graph is rainbow vertex-connected if there is a rainbow vertex path between every pair of its vertices. In the {\sc Rainbow Vertex Coloring (RVC)} problem we want to decide whether the vertices of a given graph can be colored with at most $k$ colors so that the graph becomes rainbow vertex-connected. This problem is known to be \NP-complete even in very restricted scenarios, and very few efficient algorithms are known for it. In this work, we give polynomial-time algorithms for RVC on permutation graphs, powers of trees and split strongly chordal graphs. The algorithm for the latter class also works for the strong variant of the problem, where the rainbow vertex paths between each vertex pair must be shortest paths.
We complement the polynomial-time solvability results for split strongly chordal graphs by showing that, for any fixed $p\geq 3$ both variants of the problem become \NP-complete when restricted to split $(S_3,\ldots,S_p)$-free graphs, where $S_q$ denotes the $q$-sun graph.
	\end{abstract}

	\section{Introduction}
Graph coloring is a classic problem within the field of structural and algorithmic graph theory that has been widely studied in many variants. One recent such variant was defined by Krivelevich and Yuster~\cite{KRIV10} and has received significant attention: the \emph{rainbow vertex coloring} problem. A vertex-colored graph is said to be \emph{rainbow vertex-connected} if between any pair of its vertices, there is a path whose internal vertices are colored with distinct colors. Such a path is called a \emph{rainbow path}. Note that this vertex coloring does not need to be a proper one; for instance, a complete graph is rainbow vertex-connected under the coloring that assigns the same color to every vertex. The {\sc Rainbow Vertex Coloring (RVC)} problem takes as input a graph $G$ and an integer $k$ and asks whether $G$ has a coloring with $k$ colors under which it is rainbow vertex-connected. The \emph{rainbow vertex connection number} of a graph $G$ is the smallest number of colors needed in one such coloring and is denoted $\rvc(G)$. More recently, Li~{et al.}~\cite{Li2014} defined a stronger variant of this problem by requiring that the rainbow paths connecting the pairs of vertices are also shortest paths between those pairs. In this case we say the graph is \emph{strong rainbow vertex-connected}. The analogous computational problem is called {\sc Strong Rainbow Vertex Coloring (SRVC)} and the corresponding parameter is denoted by $\srvc(G)$.

Both the \krvc~and the \ksrvc~problems are \NP-complete for every $k \geq 2$~\cite{CHEN20114531,Chen2013,EIBEN16-dam}, and remain \NP-complete even on bipartite graphs and split graphs~\cite{MFCS2018}.
Both problems are also \NP-hard to approximate within a factor of $n^{1/3-\epsilon}$ for every $\epsilon > 0$, even when restricted to bipartite graphs and split graphs~\cite{MFCS2018}. Contrasting these results, it was shown that \krvc~and \ksrvc~are linear-time solvable on bipartite permutation graphs and block graphs~\cite{MFCS2018}, and on planar graphs for every fixed $k$~\cite{LAURIPHD}. Finally, \krvc~is also known to be linear time solvable on interval graphs~\cite{MFCS2018}.

The above mentioned results on bipartite permutation graphs and interval graphs led Heggernes et al.~\cite{MFCS2018} to formulate the following conjecture concerning diametral path graphs. Recall that a graph $G$ is a \emph{diametral path graph} if every induced subgraph $H$ has a dominating path whose length equals the diameter of~$H$.

\begin{conjecture}[{Heggernes et al.~\cite[Conjecture~15]{MFCS2018}}] \label{conj:diametral}
Let $G$ be a diametral path graph. Then $\rvc(G) = \diam(G)-1$.
\end{conjecture}

In this context, it is interesting to remark that both bipartite permutation graphs and interval graphs are diametral path graphs, and that Heggernes et al.~\cite{MFCS2018} showed that the conjecture is true for these graphs.

\paragraph{Our Results}
Our main contribution is to show that the above conjecture is true for permutation graphs.

\begin{theoremintro}[=Theorem \ref{thm:permutation}] \label{thm:permutationIntro}
If $G$ is a permutation graph on $n$ vertices, then $\rvc(G)=\diam(G)-1$ and the corresponding rainbow vertex coloring can be found in $O(n^2)$ time.
\end{theoremintro}

This generalizes the earlier result on bipartite permutation graphs~\cite{MFCS2018}. The proof of our result follows from a thorough investigation of shortest paths in permutation graphs. We show that there are two special shortest paths that ensure that a rainbow vertex coloring with $\diam(G)-1$ colors can be found.

We also further the investigation of the rainbow vertex connection number of chordal graphs. As the problem is \NP-hard and hard to approximate on split graphs~\cite{MFCS2018}, the hope for polynomial-time solvability rests either within subclasses of split graphs or other chordal graphs that are not inclusion-wise related to split graphs (such as the previously studied interval graphs and block graphs~\cite{MFCS2018}). We make progress in both directions.	

First, we show that the problem is polynomial-time solvable on strongly chordal split graphs.

\begin{theoremintro}[=Theorem \ref{thm:split-strongly-chordal}] \label{thm:split-strongly-chordal-Intro}
If $G$ is a split strongly chordal graph with $\ell$ cut vertices, then $\rvc(G)=\srvc(G)=\max\{\diam(G)-1,\ell\}$.
\end{theoremintro}

In order to obtain the above result, we exploit an interesting structural property of split strongly chordal graphs. Namely, if $G$ is a split strongly chordal graph with clique $K$ and independent set $S$, there exists a spanning tree of $G[K]$ such that the neighborhood of each vertex of $S$ induces a subtree of this tree.

Second, we show that \krvc~remains polynomial-time solvable on powers of trees. This proof is based on a case analysis, depending on whether the diameter of the tree is a multiple of the power and how many long branches the tree has. We show that in some cases $\diam(G)-1$ many colors are enough to rainbow vertex color these graphs, but surprisingly this is not always the case. There are graphs in this graph class that actually require $\diam(G)$ colors in order to be rainbow vertex colored. We provide a complete characterization of such graphs, as well as a polynomial time algorithm to optimally rainbow vertex color any power of tree. 

\begin{theoremintro}[=Theorem \ref{thm:powers-of-trees}] \label{thm:powers-of-trees-Intro}
If $G$ is a power of a tree, then $\rvc(G)\in\{\diam(G)-1,\diam(G)\}$, and the corresponding optimal rainbow vertex coloring can be found in time that is linear in the size of $G$. 
\end{theoremintro}

As far as we are aware, this result provides the first graph class in which the rainbow vertex connection number is computable in polynomial time and does not always equal one of the two trivial lower bounds on the number of colors (\emph{e.g.}, $\diam(G)-1$ and the number of cut vertices).

	\section{Preliminaries}
	
	Whenever we write graph, we will mean a finite undirected simple graph. We assume throughout that all graphs are connected and have at least four vertices.
		
	Let $G= (V,E)$ be a graph. For two vertices $u,v \in V$, we use $u\sim v$ to denote that $u$ and $v$ are adjacent. For a vertex $v \in V$, we write $d_G(v)$ for its degree. For a subgraph $H$ of $G$, we write $V_H$ for the set of vertices of $H$. Specifically, for a path $P$ in $G$, we write $V_P$ for the vertices of $P$. If $X \subseteq V$, then by $G[X]$ we denote the subgraph of $G$ induced by $X$, that is, $G[X] = (X, E \cap (X \times X))$. We use $N(v) = \{ u \in V \mid u \sim v\}$ and $N[v] = N(v) \cup \{v\}$.
	
	The length of a path $P$ equals the number of edges of $P$. The distance $d_G(u,v)$ is the length of a shortest $u,v$-path in $G$. If the graph $G$ is clear from the context, we simply write $d(u,v)$. The \emph{diameter} $\diam(G)$ of $G$ is the length of the longest shortest path between two vertices in $G$, that is, $\diam(G) = \max\{d_G(u,v) \mid u,v \in V \}$. A \emph{center} of a graph $G$ is a vertex $c$ such that $\max\{d_G(c,v) \mid v \in V\}$ is maximum among all vertices of $G$. Note that a graph can have multiple centers and that a tree can have at most two.

	A graph $G$ is a \emph{permutation graph} if it is an intersection graph of line segments between two parallel lines (see Figure \ref{fig:permutation-graph}). The set of line segments that induce the permutation graph is called an \emph{intersection model}. Alternatively, if $G$ has $n$ vertices, then there is a permutation $\sigma$ of $\{1,\ldots,n\}$ such that vertex $i$ and vertex $j$ with $i < j$ are adjacent in $G$ if and only if $j$ comes before $i$ in $\sigma$.
	
	A graph $G$ is a \emph{chordal graph} if every cycle $C = \{c_1,\ldots,c_\ell\}$ on $\ell \geq 4$ vertices has a chord, meaning an edge between two non-consecutive vertices of the cycle.
	
	A graph $G$ is a \emph{split graph} if $V_G$ can be split into two sets, $K$ and $S$, such that $K$ induces a clique in $G$ and $S$ induces an independent set in $G$.
	
	For any $k \geq 3$, we denote by $S_k$ the $k$-sun on $2k$ vertices, that is, a graph with a clique $c_1,\ldots,c_k$ on $k$ vertices and an independent set $v_1,\ldots,v_k$ of $k$ vertices such that $v_i$ is adjacent to $c_i$ and $c_{i+1}$ for every $1 \leq i < k$ and $v_k$ is adjacent to $c_k$ and $c_1$. 
	A graph $G$ is a \emph{strongly chordal graph} if it is chordal and it has no induced subgraph isomorphic to a $k$-sun for any $k \geq 3$.

	The \emph{$k$-th power} of a graph $G$ for $k \geq 1$, denoted by $G^k$, is the graph on the same vertex set of $G$ where $u \sim v$ in $G^k$ if and only if there is a path of length at most $k$ from $u$ to $v$ in $G$. In particular, $G^1 = G$. If $G$ is a tree, then $G^k$ is a chordal graph for any $k \geq 1$.

	Finally, we observe the following.
	
	\begin{observation} \label{obs:basic-diam-color}
	If $\diam(G) \leq 2$, then $\srvc(G) = \rvc(G)=1$.
	\end{observation}
	\begin{proof}
	Color all vertices of $G$ by color~$1$. It suffices to note that between any two vertices, there is a shortest path with at most one internal vertex.
	\end{proof}

	\section{Permutation graphs} \label{sec:permutation}
	In this section, we consider rainbow coloring on permutation graphs. Let $G$ be a permutation graph. 
 	Let $\eL_1$ and $\eL_2$ be two parallel lines in the plane and for each $v\in V_G$, let $s_v$ be the segment associated to $v$ in the intersection model. We denote by $t(v)$ the extreme of $s_v$ in $\eL_1$, that is $t(v) = s_v \cap \eL_1$, and we refer to $t(v)$ as the top end point of $s_v$. By $b(v)$ we denote the extreme of $s_v$ in $\eL_2$, the bottom end point of $s_v$. Throughout, we assume that an intersection model is given; otherwise, one can be computed in linear time~\cite{McConnellS1999}.
	
	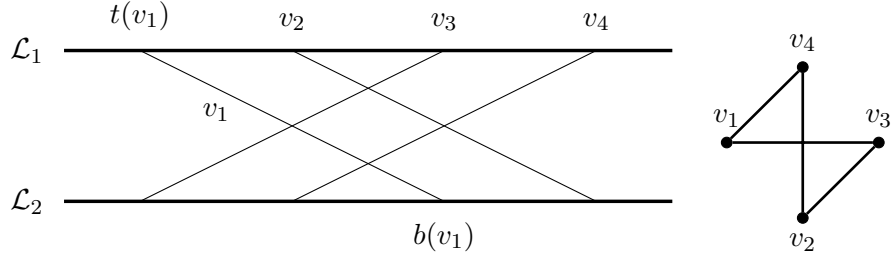
\begin{figure}
	\centering
	\begin{tikzpicture}
	\draw[line] (0,0) -- (8,0);
	\draw[line] (0,2) -- (8,2);
	
	\draw[segment] (1,2) -- (5,0);
	\draw[segment] (3,2) -- (7,0);
	\draw[segment] (5,2) -- (1,0);
	\draw[segment] (7,2) -- (3,0);
	
	\node[label = above:$v_1$] (v1) at (2,.8) {};
	\node[label = above:$t(v_1)$] (v1) at (1,2) {};
	\node[label = below:$b(v_1)$] (v1) at (5,0) {};
	\node[label = $v_2$] (k) at (3,2) {};
	\node[label = $v_3$] (k1) at (5,2) {};
	\node[label = $v_4$] (k1) at (7,2) {};
	\node[label = left:$\eL_1$] (k1) at (0,2) {};
	\node[label = left:$\eL_2$] (k1) at (0,0) {};
	\end{tikzpicture}\quad
	\begin{tikzpicture}
		\node[vertex, label = below:$v_2$] (a) at (0,0) {};
		\node[vertex, label = $v_3$] (b) at (1,1) {};
		\node[vertex, label = $v_4$] (c) at (0,2) {};
		\node[vertex, label = $v_1$] (d) at (-1,1) {};
		\draw[edge] (a) -- (b) -- (d) -- (c) --(a);
	\end{tikzpicture}
	\caption{An intersection model and the corresponding permutation graph.} \label{fig:permutation-graph}
\end{figure}
	
	Whenever we write ``$u$ intersects $v$'' for two vertices $u$ and $v$, we mean $s_u$ intersects $s_v$. 
	For two vertices $u$ and $v$, with $u\neq v$, there are several options for $u, v$ in the intersection model. 
	If \begin{align*}
	&t(u) < t(v) \text{ and } b(u) > b(v), \text{ then $u \sim v$,} \\
	&t(u) > t(v) \text{ and } b(u) < b(v), \text{ then $u\sim v$,}\\
	&t(u) < t(v) \text{ and } b(u) < b(v), \text{ then we say `$u$ is left of $v$' and write $u \prec v$,}\\
	&t(u) > t(v) \text{ and } b(u) > b(v), \text{ then we say `$u$ is right of $v$' and write $u \succ v$.}
	\end{align*} 
	We also use the notation $u \preceq v$ if $t(u) \leq t(v)$  and $b(u) \leq b(v)$. Notice that `$\prec$' is a partial ordering on the vertices of the graph, in particular, $u \nsucc v$ does not imply that $u \preceq v$. 
	
	For each pair $u,v \in V(G)$, Mondal et al.~\cite{MondalPP2003} define two $u$-$v$ paths, one of which is shortest.
	Define a path $X_{u,v}$ as follows. If $u \sim v$, $X_{u,v}$ will be $u, v$. Otherwise, assume without loss of generality that $u \prec v$. Start with $x_1 = u$. Of all vertices $x$ that intersect $u$ with $t(x) > t(u)$, let $x_2$ be the one with largest $t(x_2)$. If there is no vertex $x$ that intersects $u$ with $t(x) > t(u)$, we say that the path $X_{u,v}$ does not exist. Otherwise, define $x_i$, with $i\geq 3$, as follows. If $x_{i-1}$ is incident to $v$, set $x_i = v$ and end the path. Otherwise, if $i$ is even (resp.\ odd), let $x_i$ be the vertex that intersects $x_{i-1}$ where $t(x_i)$ (resp.\ $b(x_i)$) is largest. 
	Notice that it cannot be that $x_{i-2}=x_i$, or $G$ would not be connected. 
	
	Analogously, we define the path $Y_{u,v}$. This path starts with $y_1 = u$. Let $y_2$ be the vertex that intersects $u$ with largest $b(y_2)$, if $b(y_2) > b(u)$ (otherwise the path $Y_{u,v}$ does not exist). Now the next vertex $y_i$ is the vertex that intersects $y_{i-1}$ with largest $b(y_i)$ (resp.\ $t(y_i)$) if $i$ is even (resp.\ odd). 
	Notice that it cannot be that $y_{i-2}=y_i$, or $G$ would not be connected. 
	
	The paths we just defined satisfy the following property. Let $z_1, z_2, z_3, \ldots, z_a$ be a path. For all $2\leq i\leq a$,
	\begin{align}
	&t(z_i) > t(z_{i-1}) \text{ and } b(z_i) < b(z_{i-1}) \text{ if $i$ is even,} \label{eq:x-even}\\ 
	&t(z_i) < t(z_{i-1}) \text{ and } b(z_i) > b(z_{i-1}) \text{ if $i$ is odd,} \label{eq:x-odd}
	\end{align}
	or, for all $2\leq i\leq a$, 
	\begin{align}
	&t(z_i) < t(z_{i-1}) \text{ and } b(z_i) > b(z_{i-1}) \text{ if $i$ is even,}\label{eq:y-even} \\ 
	&t(z_i) > t(z_{i-1}) \text{ and } b(z_i) < b(z_{i-1}) \text{ if $i$ is odd.}\label{eq:y-odd}
	\end{align}
	Note that Equations~\eqref{eq:x-even} and~\eqref{eq:x-odd} hold for $X_{u,v}$, by definition, and that Equations~\eqref{eq:y-even} and~\eqref{eq:y-odd} hold for $Y_{u,v}$, by definition. In later proofs we will often use this property.

	\begin{lemma}
	\label{lem:X-Y-shortest}
		$X_{u,v}$ or $Y_{u,v}$ is a shortest $u,v$-path. 
	\end{lemma}
	\begin{proof}
		See [{Mondal et al.~\cite[Lemma~5]{MondalPP2003}}], or see Appendix \ref{app:shortest-paths}. 
	\end{proof}
	
	We define two special paths $P$ and $Q$. For the definition of $P$, let $p_1$ be the vertex such that $t(p_1)$ is smallest among all vertices of $G$. Perform the same process as in the construction of $X_{p_1,\cdot}$: for $i\geq 2$, if $i$ is even (resp.\ odd), let $p_i$ be the vertex that intersects $p_{i-1}$ where $t(p_i)$ (resp.\ $b(p_i)$) is largest. Let $P$ denote the resulting path and let $p_d$ denote the last vertex of $P$. Observe that $P = X_{p_1,p_d}$. 
	
	For the definition of $Q$, let $q_1$ be the vertex such that $b(q_1)$ is smallest among all vertices of $G$. Perform the same process as in the construction of $Y_{q_1,\cdot}$: for $i\geq 2$, if $i$ is even (resp.\ odd), let $q_i$ be the vertex that intersects $q_{i-1}$ where $b(q_i)$ (resp.\ $t(q_i)$) is largest. Let $Q$ denote the resulting path and let $q_{d'}$ denote the last vertex of $Q$. Observe that $Q = Y_{q_1,q_{d'}}$.
	
	\begin{corollary} \label{cor:P-Q-shortest}
		$P$ is a shortest $p_1, p_d$-path and $Q$ is a shortest $q_1, q_{d'}$-path.
	\end{corollary}
	\begin{proof}
		By Lemma~\ref{lem:X-Y-shortest}, it follows that $X_{p_1, p_d}$ or $Y_{p_1, p_d}$ is a shortest $p_1,p_d$-path. We claim that $Y_{p_1, p_d}$ does not exist. Assume the contrary and let $y_2$ be the vertex that follows $p_1$ on $Y_{p_1, p_d}$. Then $b(y_2) > b(p_1)$ by the definition and existence of $Y_{p_1, p_d}$. Since $y_2$ intersects $p_1$, it follows that $t(y_2) < t(p_1)$. But this contradicts the definition of $p_1$. Hence, $Y_{p_1, p_d}$ does not exist. Therefore, $X_{p_1, p_d} = P$ is a shortest $p_1,p_d$-path. Analogously, we can prove that $Q$ is a shortest $q_1, q_{d'}$-path. 
	\end{proof}
		
	We will prove some more useful properties about the paths $P$ and $Q$, before we show the rainbow coloring. 
	
	\begin{lemma} \label{lem:P-Q-end}
		Let $v_t$, resp.\ $v_b$, be the segment that has the rightmost top, resp.\ bottom, end point. 
		Then $p_d = v_t$ and $p_{d-1} = v_b$ if $d$ is even, and vice versa if $d$ is odd. Furthermore, we have $q_{d'} = v_b$ and $q_{d'-1} = v_t$ if $d'$ is even, and vice versa if $d'$ is odd.
	\end{lemma}
	\begin{proof}
		Suppose that $p_d = v_t$. Since the top of $p_d$ is rightmost, we see that $d$ is even. Then, by definition of $P$, it holds that the segment, of all segments that intersect $p_d$, that ends rightmost at the bottom equals $p_{d-1}$. So, $p_{d-1} = v_b$. 
		
		Suppose that $p_d \neq v_t$ and $v_t \sim p_{d}$. It is clear that $t(p_{d}) < t(v_t)$, thus $b(p_{d}) > b(v_t)$. If $d$ is even, then $v_t \sim p_{d-1}$. But $p_{d}$ is the vertex that intersects $p_{d-1}$ which has the rightmost top end, yielding a contradiction with the fact that $t(p_{d}) < t(v_t)$. So $d$ is odd. 
		Since the path ends at $p_d$, it follows that the vertex that intersects $p_d$ that has the rightmost top is $p_{d-1}$. Hence $p_{d-1} = v_t$. 
		By definition, $p_d$ is the vertex that intersects $p_{d-1}$ that has the rightmost bottom. Thus $p_{d} = v_b$.
		
		Suppose that $p_d \neq v_t$ and $v_t \nsim p_{d}$. It holds that $p_d$ is left of $v_t$. Assume that $d$ is even. Consider the shortest $p_d, v_t$-path. We know that either $X_{p_d, v_t}$ or $Y_{p_d, v_t}$ is a shortest path. (See Figure \ref{fig:P-Q-end}.) Suppose that $X_{u,v}$ is the shortest. We know that $t(x_2) > t(p_d)$ and $b(x_2) < b(p_d)$. It follows that $x_2$ intersects $p_{d-1}$. This yields a contradiction with the choice of $p_d$. Suppose that $Y_{u,v}$ is the shortest path. By the definitions of $y_2$ and $P$, it holds that $y_2 = p_{d-1}$. Then $y_3 = p_d = y_1$, so $Y_{u,v}$ is not a shortest path. We conclude that either $p_d = v_t$ or $v_t \sim p_{d}$. The case for 
		$d$ odd is analogous. 
		\begin{figure}
	\centering
	\begin{tikzpicture}
	\draw[line] (0,0) -- (7,0);
	\draw[line] (0,2) -- (7,2);
	
	\draw[segment] (1,2) -- (4,0);
	\draw[segment] (2,0) -- (4,2);
	\draw[segment] (5,0) -- (6,2);
	\draw[possiblesegment] (5,2) -- (1,0);
	
	\node[label = $p_{d-1}$] (k) at (1,2) {};
	\node[label = $p_{d}$] (k1) at (4,2) {};
	\node[label = $v_t$] (k) at (6,2) {};
	\node[label = $x_2$] (v) at (5,2) {};
	\end{tikzpicture}\quad
	\begin{tikzpicture}
	\draw[line] (0,0) -- (7,0);
	\draw[line] (0,2) -- (7,2);
	
	\draw[segment] (1,2) -- (4,0);
	\draw[segment] (2,0) -- (4,2);
	\draw[segment] (5,0) -- (6,2);
	
	\node[label = {$p_{d-1} = y_2$}] (k) at (1,2) {};
	\node[label = $p_{d}$] (k1) at (4,2) {};
	\node[label = $v_t$] (k) at (6,2) {};
	\end{tikzpicture}
	\caption{See Lemma \ref{lem:P-Q-end}, case $p_d \neq v_t$, $v_t \nsim p_d$. Either $X_{p_d, v_t}$ or $Y_{p_d,v_t}$ is a shortest path. } \label{fig:P-Q-end}
\end{figure}
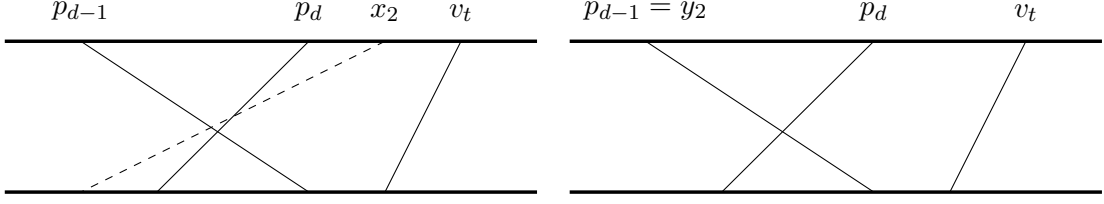
		
		The proof for $Q$ is analogous. 
	\end{proof}

	\begin{lemma} \label{lem:P-Q-dominating}
		The sets $V_P\setminus \{p_d\}$ and $V_Q \setminus\{q_{d'}\}$ are dominating sets of $G$.
	\end{lemma}
	
	\begin{proof}
		Suppose that $V_P\setminus \{p_d\}$ is not a dominating set, and let $v$ be a vertex that is not dominated. We prove with induction that $v$ is right of $p_i$ for all $1\leq i \leq d-1$. 
		
		It is clear that $v \succ p_1$, since $v \nsim p_1$ and $t(v) > t(p_1)$. 
		Suppose that $v \succ p_{k-1}$, for some $k$. Suppose that $k$ is even. We use the induction hypothesis and Equation \eqref{eq:x-even}, to see that $b(v) > b(p_{i-1}) > b(p_i)$. Since $v \nsim p_i$, it follows that $v\succ p_i$. 
		Suppose that $k$ is odd. Then $t(v) > t(p_{i-1}) > t(p_i)$, by the induction hypothesis and Equation \eqref{eq:x-odd}. Since $v \nsim p_i$, it follows that $v\succ p_i$.
		
		By Lemma \ref{lem:P-Q-end}, we know that $p_{d-1} = v_t$ or $p_{d-1} = v_b$. In both cases, there is no vertex $v$ with $v \succ p_{d-1}$. This yields a contradiction, thus $V_P\setminus \{p_d\}$ is a dominating set. 
		
		The proof for $V_Q \setminus\{q_{d'}\}$ is analogous. 
	\end{proof}
	
	\begin{lemma} \label{lem:length}
		It holds that $d = \diam(G)$ or $d= \diam(G) +1$, and $d' = \diam(G)$ or $d' = \diam(G) +1$.
	\end{lemma}
	\begin{proof}
		Since $P$ is a shortest path (by Corollary \ref{cor:P-Q-shortest}), we have that $d-1 \leq \diam(G)$. Or, equivalently $d \leq \diam(G) + 1$. 
		
		Let $u$ and $v$ be two vertices. Since the set $V_P \setminus \{p_d\}$ is a dominating set (see Lemma \ref{lem:P-Q-dominating}), there are $p_i, p_j \in V_P \setminus \{p_d\}$ such that $p_i \sim u$, $p_j \sim v$. Without loss of generality, we assume that $i\leq j$. Then the path $u, p_i, p_{i+1}, \ldots, p_j, v$ has length at most $d$. Thus $\diam(G) \leq d$. 
		
		The proof for $d'$ is analogous. 
	\end{proof}

	Now we start a breadth-first search from $p_1$. Call the layers $L_1, L_2, \ldots, L_r$. Since $P$ is a shortest path, it follows that $p_i \in L_i$ for every $i$. Thus $r\geq d$. Since $V_P\setminus \{p_d\}$ is a dominating set, we conclude that $r = d$, thus the layers of the breadth-first search are $L_1, L_2, \ldots, L_d$. 
	We also start a breadth-first search in $q_1$, and call the layers $M_1, M_2, \ldots, M_{d'}$. Again, we have that $q_i \in M_i$ for every $i$.  
	
	A nice property of the path $P$ is that every vertex $p_i$ is adjacent to all vertices in the next layer $L_{i+1}$.

	\begin{lemma} \label{lem:neighbourhood}
		For every $i$, it holds that $L_{i+1} \subseteq N(p_i)$ and $M_{i+1} \subseteq N(q_i)$. 
	\end{lemma}
	\begin{proof}
		We will prove a somewhat stronger result by induction, namely that $L_{i+1} \subseteq N(p_i)$ and if $i$ is even (resp.\ odd) we have that for every $u$ in $L_{i+1}$: \begin{equation} \label{eq:neighbourhood}
			t(u) < t(p_i) \text{ and } b(u) > b(p_i) \text{ (resp.\ $t(u) > t(p_i)$ and $b(u) < b(p_i)$).}  
		\end{equation}
		
		We use a proof by induction. The first layer $L_1$ contains only $p_1$. It is clear that every vertex in the second layer $L_2$ is a neighbour of $p_1$. Moreover, by the definition of $p_1$, we have that $t(u) > t(p_1)$ for all $u\in L_2$, and thus $b(u) < b(p_1)$. 
		
		Suppose that $L_{i+1} \subseteq N(p_{i})$ and Equation \eqref{eq:neighbourhood} holds for every $i<k$. Let $v$ be a vertex in $L_{k+1}$. We know that $v$ does not intersect $p_{k-1}$, otherwise $v$ would be contained in $L_k$. So we know that $t(v) > t(p_{k-1})$ and $b(v) > b(p_{k-1})$. Since $v$ is in layer $L_{k+1}$, $v$ intersects $u$ for some $u\in L_{k}$ (see Figure \ref{fig:neighbourhood}). 
		So we either have that $t(v) < t(u)$ and $b(v)> b(u)$ or we have $t(v) > t(u)$ and $b(v) < b(u)$. If $k$ is even (resp.\ odd), we have $t(v) < t(u)$ and $b(v)> b(u)$ (resp.\ $t(v) > t(u)$ and $b(v) < b(u)$), otherwise, by the induction hypothesis, $v$ would intersect $p_{k-1}$. By the induction hypothesis $u$ intersects $p_{k-1}$, so by the definition of $p_k$, we have that $t({p_k}) \geq t(u)$ (resp.\ $b(p_k) \geq b(u)$). It follows that $t(v) < t(p_k)$ (resp.\ $b(v) < b(p_{k})$). We know that $b(p_k) < b(p_{k-1})$ (resp.\ $t(p_k) < t(p_{k-1})$), thus $b(v) > b(p_k)$ (resp.\ $t(v) > t(p_{k})$). We conclude that $v$ intersects $p_k$, and $t(v) < t(p_k)$ and $b(v) > b({p_k})$ (resp.\ $t(v) > t({p_k})$ and $b(v) < b({p_k})$). 
		
		So $L_{k+1} \subseteq N(p_k)$ for every $k$. The proof that $M_{k+1} \subseteq N(q_k)$ is analogous. 
	\end{proof}
	\begin{figure}
	\centering
	\begin{tikzpicture}
		\draw[line] (0,0) -- (10,0);
		\draw[line] (0,2) -- (10,2);
		
		\draw[segment] (1,2) -- (6,0);
		\draw[segment] (4,2) -- (3,0);
		\draw[segment] (7,2) -- (4,0);
		\draw[segment] (6,2) -- (5,0);
		\draw[segment] (5,2) -- (9,0);
		
		\node[label = $p_{k-1}$] (k) at (1,2) {};
		\node[label = $p_{k}$] (k1) at (7,2) {};
		\node[label = $u$] (u) at (6,2) {};
		\node[label = below:$v$] (v) at (9,0) {};
	\end{tikzpicture}
	\caption{If a vertex $v$ in layer $L_{k+1}$ intersect a vertex $u$ in layer $L_{k}$, then it also intersects $p_{k}$. See Lemma \ref{lem:neighbourhood}.} \label{fig:neighbourhood}
\end{figure}
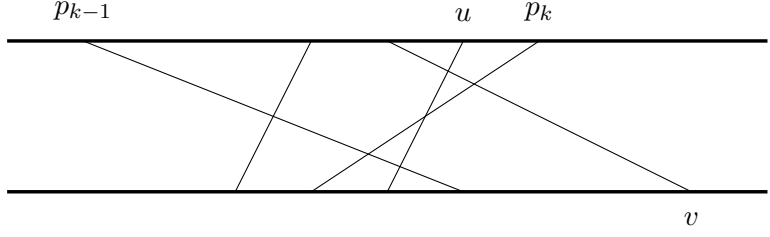
	
	For an illustration of the structure of $G$, see Figure \ref{fig:layers}. If $d=\diam(G)$ or $d' = \diam(G)$, we will color $G$ layer by layer to obtain a rainbow coloring. 
	
	\begin{lemma} \label{lem:layers}
		If $d = \diam(G)$ or $d' = \diam(G)$, then $\rvc(G) = \diam(G)-1$. 
	\end{lemma}
	\begin{proof}
		Assume that $d = \diam(G)$. Consider the following coloring (see Figure \ref{fig:layers}):
		\begin{align*}
		c(v) = \begin{cases}
		i & \text{if $v\in L_i$, $1\leq i\leq d-1$,}\\
		1 & \text{otherwise.}
		\end{cases}
		\end{align*}
		\begin{figure}
	\centering
	\begin{tikzpicture}
\ellips{0}{0}
\node[vertex, label=$p_1$] (p) at (-2,0) {};
\node (l) at (-2,-1.5) {$L_1$};
\ellips{2}{0}
\node[vertex, label=$p_2$] (p) at (0,0) {};
\node (l) at (0,-1.5) {$L_2$};
\ellips{4}{0} 
\node[vertex, label=$p_3$] (p) at (2,0) {};
\node (l) at (2,-1.5) {$L_3$};

\node (d) at (5,0) {$\ldots$};

\draw[fill=gray!30,opacity=0.5](6,0)circle[x radius=.5,y radius=1];
\ellips{8}{0}
\node[vertex, label=$p_{d-1}$] (p) at (6,0) {};
\node (l) at (6,-1.5) {$L_{d-1}$};
\node[vertex, label=$p_d$] (p) at (8,0) {};
\node (l) at (8,-1.5) {$L_d$};
\end{tikzpicture}
\caption{The layers of the BFS, the layers $L_1, L_2, \ldots, L_{d-1}$ are all assigned a different color. See Lemma \ref{lem:layers}} \label{fig:layers}
\end{figure}
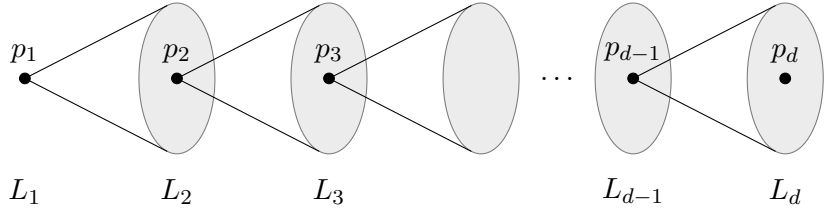
		This coloring uses $d-1 = \diam(G) - 1$ colors. We claim that it is a rainbow coloring. Let $u$ and $v$ be two vertices. Then $u \in L_i$, $v\in L_j$ for some $i, j$. Without loss of generality, assume that $i\leq j$. If $i = 1$, then $u = p_1$. Then, by Lemma \ref{lem:neighbourhood}, the path $p_1, p_2, \ldots, p_{j-1}, v$ is a rainbow path. If $i>1$, again by Lemma \ref{lem:neighbourhood}, the path $u, p_{i-1}, p_i, \ldots, p_{j-1}, v$ is a rainbow path. 
		
		We conclude that $\rvc(G) = \diam(G)-1$. The proof for $d' = \diam(G)$ is analogous. 
	\end{proof}

	Consider the case where $d = d' = \diam(G) + 1$. In this case, we will still color the layers of a breadth-first search that starts at $p_1$, but we have to reuse the color of the first layer for layer $L_{d-1}$. 
	We consider the coloring:
	\begin{align*}
	c(v) = \begin{cases}
	i & \text{if $v\in L_i$, $1\leq i\leq d-2$,}\\
	1 & \text{if $v\in L_{d-1}$,} \\
	2 & \text{if $v\in L_{d}$.}
	\end{cases}
	\end{align*}
	We will show that this is a rainbow coloring. For almost every $u$ and $v$, we readily construct a rainbow path.  
	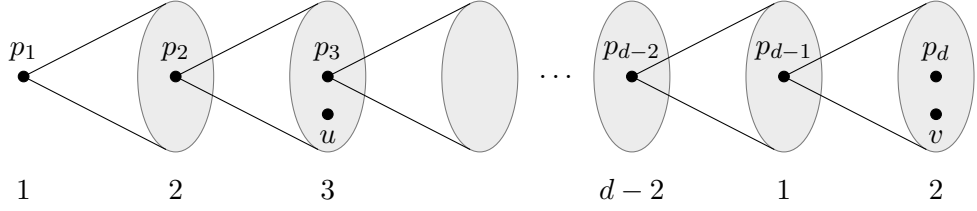
\begin{figure}
	\centering
	\begin{tikzpicture}
\ellips{0}{0}
\node[vertex, label=$p_1$] (p) at (-2,0) {};
\node (l) at (-2,-1.5) {$1$};
\ellips{2}{0}
\node[vertex, label=$p_2$] (p) at (0,0) {};
\node (l) at (0,-1.5) {$2$};
\ellips{4}{0} 
\node[vertex, label=$p_3$] (p) at (2,0) {};
\node (l) at (2,-1.5) {$3$};

\node (d) at (5,0) {$\ldots$};

\draw[fill=gray!30,opacity=0.5](6,0)circle[x radius=.5,y radius=1];
\ellips{8}{0}
\ellips{10}{0}
\node[vertex, label=$p_{d-2}$] (p) at (6,0) {};
\node (l) at (6,-1.5) {$d-2$};
\node[vertex, label=$p_{d-1}$] (p) at (8,0) {};
\node (l) at (8,-1.5) {$1$};
\node[vertex, label=$p_d$] (p) at (10,0) {};
\node (l) at (10,-1.5) {$2$};

\node[vertex, label=below:$u$] (u) at (2,-.5) {};
\node[vertex, label=below:$v$] (v) at (10,-.5) {}; 
\end{tikzpicture}
\caption{The numbers indicate the colors of the layers. The path $u, p_2, p_3, \ldots, p_{d-1}, v$ is a rainbow path. See Lemma \ref{lem:almost-all-paths-rainbow}} \label{fig:almost-all-paths-rainbow}
\end{figure}
	
	\begin{lemma} \label{lem:almost-all-paths-rainbow}
		For the following $u$ and $v$, there exists a rainbow path: 
		\begin{enumerate}
			\item for $u = p_1$, and $v$ arbitrary,
			\item for $u \in L_i$ with $i \geq 3$, and $v \in L_j$ with $j\geq 3$, 
			\item for $u \in L_2$, and $v\notin L_{d}$,
			\item for $u \in L_2$, $u \sim p_2$, and $v \in L_d$.
		\end{enumerate}	
	\end{lemma}
	\begin{proof}
		\begin{enumerate}
			\item Suppose that $v\in L_i$. By Lemma \ref{lem:neighbourhood}, we know that $v \sim p_{i-1}$. The path $p_1, p_2, \ldots, p_{i-1}, v$ is a rainbow path. 
			\item Without loss of generality, we assume that $i\leq j$. By Lemma \ref{lem:neighbourhood}, we know that $u \sim p_{i-1}$ and $v\sim p_{j-1}$. Consider the path $u, p_{i-1}, p_{i}, \ldots, p_{j-1}, v$. The internal vertices are in different layers between $L_2$ and $L_{d-1}$. Hence, this is a rainbow path. 
			\item Suppose that $v\in L_i$, with $i<d$. Again, by Lemma \ref{lem:neighbourhood}, we know that $v \sim p_{i-1}$. The path $u, p_1, p_2, \ldots, p_{i-1}, v$ is a rainbow path.
			\item Notice that the path $u, p_2, p_3, \ldots, p_{d-1}, v$ is a rainbow path. 	
		\end{enumerate}	
	\end{proof}
	
	There are some vertices $u$ and $v$, for which we did not yet construct a rainbow path. The case that is left, is the following: 
	\begin{enumerate}\setcounter{enumi}{4}
		\item for $u \in L_2$, $u\nsim p_2$, and $v \in L_d$.
	\end{enumerate}	
	The path via $P$, $u, p_1, p_2, \ldots, p_{d-1}, v$, does not suffice in this case, because it uses $p_1$ and $p_{d-1}$, which are both colored with color $1$. So this is not a rainbow path. For some cases we will show that a similar path via $Q$ is a rainbow path. For other cases, we can use the shortest $u,v$-path, that is, $X_{u,v}$ or $Y_{u,v}$ is a rainbow path.
	
	\begin{lemma} \label{lem:u-intersects-q1}
		If $u \in L_2$, $u \nsim p_2$, then $u \sim q_1$ or $u \sim q_2$. 
	\end{lemma}
	\begin{proof}
		Since $u \in L_2$, we know that $u$ intersects $p_1$. Hence $t(u) > t(p_1)$ and $b(u) < b(p_1)$. 
		If $u \nsim q_1$, then $u$ is right of $q_1$, since $q_1$ has the leftmost bottom end. So $t(u) > t(q_1) > t(q_2)$ by Equation \eqref{eq:y-even}. Since $p_1$ and $q_2$ both intersect $q_1$, we now that $b(q_2) \geq b(p_1)$, by the definition of $q_2$. Thus $b(u) < b(p_1) \leq b(q_2)$. We conclude that $u$ intersects $q_2$. 
	\end{proof}
	
	\begin{lemma} \label{lem:P-Q-same-end}
		If $d = d' = \diam(G) + 1$, then $p_d = q_{d-1}$ and $p_{d-1} = q_d$. 
	\end{lemma}
	\begin{proof}
		By Lemma \ref{lem:P-Q-end}, we know that, if $d$ is even, then $p_d = v_t$ and $p_{d-1} = v_b$, and conversely if $d$ is odd. For $Q$ we have that $q_d = v_b$ and $q_{d-1} = v_t$ if $d$ is even, and conversely if $d$ is odd. 
		We conclude that $p_d = q_{d-1}$ and $p_{d-1} = q_d$. 
	\end{proof}

	\begin{corollary} \label{cor:q-in-layer}
		If $d = d' = \diam(G) + 1$, it holds that $q_i \in L_{i+1}$, for every $1 \leq i < d$.
	\end{corollary}
	\begin{proof}
		It is clear that $q_1 \in L_2$, since $q_1 \sim p_1$. By Lemma \ref{lem:P-Q-same-end}, we know that $q_{d-1} = p_d$, so $q_{d-1} \in L_{d}$. Since $q_1, q_2, \ldots, q_{d-1}$ is a path of length $d-1$, it follows that $q_i \in L_{i+1}$. 
	\end{proof}

	\begin{lemma} \label{lem:v-intersects-pd}
		If $d = d' = \diam(G) + 1$, then for every $v \in L_{d}$, if $v \nsim q_{d-2}$, then $v\sim q_{d-1}$.
	\end{lemma}
	\begin{proof}
		Since $v \in L_d$, we know that $v \sim p_{d-1}$, so by Lemma \ref{lem:P-Q-same-end}, $v \sim q_{d}$. (See Figure \ref{fig:v-intersects-pd}.)
		Assume that $d$ is even. Then $q_{d} = v_b$ by Lemma~\ref{lem:P-Q-end}. It follows that $b(v)< b(v_b)$ and $t(v) > t(v_b)$. If $v \nsim q_{d-2}$, we know that $b(v) > b(q_{d-2}) > b(q_{d-1})$. Since $q_{d-1}= v_t$ by Lemma~\ref{lem:P-Q-end}, we know that $t(v) < t(q_{d-1})$. So, $v$ intersects $q_{d-1}$. 
		
		The case that $d$ is odd is analogous. 
	\end{proof}
	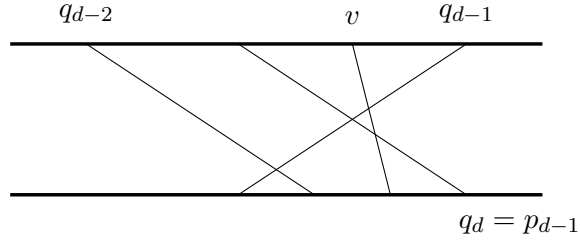
\begin{figure}
	\centering
	\begin{tikzpicture}
		\draw[line] (0,0) -- (7,0);
		\draw[line] (0,2) -- (7,2);
		
		\draw[segment] (1,2) -- (4,0);
		\draw[segment] (6,2) -- (3,0);
		\draw[segment] (3,2) -- (6,0);
		\draw[segment] (4.5,2) -- (5,0);
		
		\node[label = $q_{d-2}$] (k) at (1,2) {};
		\node[label = below:{$q_{d} = p_{d-1}$}] (k1) at (6.7,0) {};
		\node[label = $q_{d-1}$] (u) at (6,2) {};
		\node[label = $v$] (v) at (4.5,2) {};
	\end{tikzpicture}
	\caption{If a vertex $v$ in layer $L_{d}$ does not intersect $q_{d-2}$, then it intersects $q_{d-1}$. See Lemma \ref{lem:v-intersects-pd}.} \label{fig:v-intersects-pd}
\end{figure}
	
	Now we can prove for even more vertices $u$ and $v$ that there is a rainbow path from $u$ to $v$, see also Figure \ref{fig:even-more-rainbow-paths}. 
	
	\begin{lemma} \label{lem:even-more-rainbow-paths}
		For the following vertices $u$ and $v$, there is a rainbow path: 
		\begin{enumerate}
			\item[5a.] for $u \in L_2$, $u\nsim p_2$, $v \in L_d$, and $v \sim q_{d-2}$,
			\item[5b.] for $u \in L_2$, $u\nsim p_2$, $v \in L_d$, and $v \nsim q_{d-2}$, $u \sim q_2$.
		\end{enumerate}
	\end{lemma}
	\begin{proof}
		\begin{itemize}
			\item[5a.] By Lemma \ref{lem:u-intersects-q1}, we know that $u \sim q_1$ or $u \sim q_2$. By Corollary \ref{cor:q-in-layer}, we know that $q_1, q_2, \ldots, q_{d-2}$ are in layers $L_2, L_3, \ldots, L_{d-1}$, each vertex in a different layer. So, either $u, q_1, q_2, \ldots, q_{d-2}, v$ or $u, q_2, q_3, \ldots, q_{d-2}, v$ is a rainbow path. 
			\item[5b.] By Lemma \ref{lem:v-intersects-pd}, we know that $v \sim q_{d-1}$. By Corollary \ref{cor:q-in-layer}, we know that $q_1, q_2, \ldots, q_{d-2}$ are in layers $L_2, L_3, \ldots, L_{d-1}$, each vertex in a different layer. The path $u, q_2, q_3, \ldots, q_{d-1}, v$ is a rainbow path. 
		\end{itemize}
	\end{proof}
	\begin{figure}
	\centering
	\begin{tikzpicture}
\ellips{0}{0}
\node[vertex, label=$p_1$] (p) at (-2,0) {};
\node (l) at (-2,-1.5) {$1$};
\ellipshigh{2}{0}
\node[vertex, label=$p_2$] (p) at (0,0.8) {};
\node (l) at (0,-1.5) {$2$};
\ellipshigh{4}{0} 
\node[vertex, label=$p_3$] (p) at (2,0.8) {};
\node (l) at (2,-1.5) {$3$};

\node (d) at (5,0) {$\ldots$};

\draw[fill=gray!30,opacity=0.5](6,0)circle[x radius=.5,y radius=1];
\ellipshigh{8}{0}
\ellipshigh{10}{0}
\node[vertex, label=$p_{d-2}$] (p) at (6,0.8) {};
\node (l) at (6,-1.5) {$d-2$};
\node[vertex, label={$p_{d-1}$}] (p) at (8,0.8) {};
\node (l) at (8,-1.5) {$1$};
\node[vertex, label=right:{$p_d = q_{d-1}$}] (p) at (10,0.8) {};
\node (l) at (10,-1.5) {$2$};

\node[vertex, label=$q_1$] (q1) at (0,0) {};
\node[vertex, label=below:$q_2$] (q2) at (2,0) {}; 
\node[vertex, label=below:$q_{d-3}$] (q3) at (6,0) {};
\node[vertex, label=below:$q_{d-2}$] (q) at (8,0) {}; 

\node[vertex, label=below:$u$] (u) at (0,-.5) {};
\node[vertex, label=below:$v$] (v) at (10,-.5) {}; 
\end{tikzpicture}
\caption{The numbers indicate the colors of the layers. If $u\sim q_1$ and $v \sim q_{d-1}$, the path $u, q_1, p_2, \ldots, q_{d-2}, v$ is a rainbow path. If $u \sim q_2$ and $v \nsim q_{d-2}$, the path $u, q_2, p_3, \ldots, q_{d-1}, v$ is a rainbow path. See Lemma \ref{lem:even-more-rainbow-paths}} \label{fig:even-more-rainbow-paths}
\end{figure}
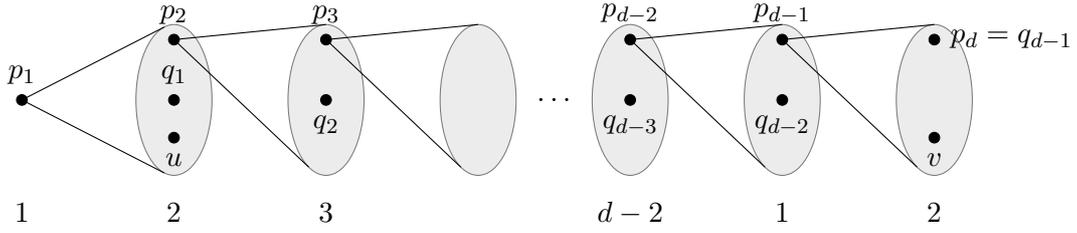

	Now there is still one case of vertices $u$ and $v$ for which we did not prove yet that there is a rainbow path. Namely:
	\begin{itemize}
		\item[5c.]  for $u \in L_2$, $u\nsim p_2$, $v \in L_d$, and $v \nsim q_{d-2}$, $u \nsim q_2$.
	\end{itemize}	
	For this last case we will show that either $X_{u,v}$ or $Y_{u,v}$ is a rainbow path.

	\begin{lemma} \label{lem:u-left-of-p2}
		If $u \sim p_1$, $u\nsim p_2$, $u\nsim q_2$, then $u \prec p_2$. 
	\end{lemma}
	\begin{proof}
		Since $u \nsim p_2$, we know that either $u \prec p_2$ or $u \succ p_2$. Suppose that $u\succ p_2$. Then $t(u) > t(p_2)$. By the definition of $p_2$, we know that $t(p_2) \geq t(q_1)$ and by Equation \eqref{eq:y-even} we see that $t(q_1) > t(q_2)$. Altogether, we see that $t(u) > t(q_2)$. By the definition of $q_2$, we know that $b(q_2) \geq b(p_1)$. Since $u\sim p_1$, it holds that $b(u) < b(p_1)$. Thus $b(u) < b(q_2)$. This yields a contradiction with the fact that $u \nsim q_2$. We conclude that $u \prec p_2$. 
	\end{proof}

	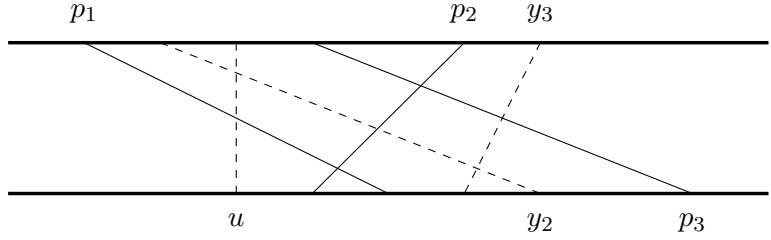
\begin{figure}
	\centering
	\begin{tikzpicture}
		\draw[line] (0,0) -- (10,0);
		\draw[line] (0,2) -- (10,2);
		
		\draw[segment] (1,2) -- (5,0);
		\draw[segment, dashed] (2,2) -- (7,0);
		\draw[segment, dashed] (3,2) -- (3,0);
		\draw[segment] (4,2) -- (9,0);
		\draw[segment] (4,0) -- (6,2);
		\draw[segment, dashed] (7,2) -- (6,0);
		
		\node[label = $p_{1}$] (k) at (1,2) {};
		\node[label = below:$u$] (u) at (3,0) {};
		\node[label = below:{$y_{2}$}] (k1) at (7,0) {};
		\node[label = $p_2$] (u) at (6,2) {};
		\node[label = below:$p_3$] (u) at (9,0) {};
		\node[label = $y_3$] (v) at (7,2) {};
	\end{tikzpicture}
	\caption{Vertex $y_2$ is in layer $L_3$. See Lemma \ref{lem:last-rainbow-paths}.} \label{fig:last-rainbow-paths}
\end{figure}
	\begin{lemma} \label{lem:last-rainbow-paths}
		For $u$ and $v$ satisfying case 5c, there is a rainbow path.
	\end{lemma}
	\begin{proof}
		We distinguish two cases, based on Lemma~\ref{lem:X-Y-shortest}: either $X_{u,v}$ is a shortest $u,v$-path or $Y_{u,v}$ is a shortest $u,v$-path.
		
		Suppose that $X_{u,v}$ is a shortest $u,v$-path. Notice that $X_{u,v}$ has at least one vertex in every layer $L_2, L_3, \ldots, L_d$. Since $X_{u,v}$ has length at most $d-1$, there is at most one layer which contains two vertices of $X_{u,v}$. It is clear that $x_1 = u$ is in layer $L_2$. We will show that $x_2$ is in layer $L_2$ as well. By definition of $x_2$, we have $t(x_2) > t(u)$ and $b(x_2) < b(u)$. Since $u \sim p_1$ and $p_1$ has the leftmost top end, we know that $t(u) > t(p_1)$ and $b(u) < b(p_1)$. We conclude that  $t(x_2) > t(p_1)$ and $b(x_2) < b(p_1)$, thus $x_2 \sim p_1$. So we see that $x_1$ and $x_2$ are both in layer $L_2$, so all internal vertices of $X_{u,v}$ are in different layers. So $X_{u,v}$ is a rainbow path.   		
		
		Suppose that $Y_{u,v}$ is a shortest $u,v$-path. Write $Y_{u,v} = u, y_2, y_3, \ldots, y_{\alpha-1}, v$. Then $\alpha = d$ or $\alpha = d-1$; note that $\alpha \leq \diam(G)+1 = d$ and that $d-1 \leq \alpha$ because $Y_{u,v}$ contains a vertex from every layer. We prove by induction that $y_i \in L_{i+1}$ for all $2 \leq i\leq \alpha - 1$. 
		
		Since $y_2$ and $p_1$ both intersect $u$, by the definition of $y_2$, it follows that $b(y_2) \geq b(p_1)$. See Figure \ref{fig:last-rainbow-paths}. If $y_2 = p_1$, then $y_{u,v} = u, p_1, p_2, \ldots, p_{d-1}, v$. Notice that the length of this path is $d = \diam(G) + 1$. This yields a contradiction with the fact that $Y_{u,v}$ is a shortest $u,v$-path. Hence, $y_2 \neq p_1$, and $b(y_2) > b(p_1)$. Since $p_1$ is the vertex with the leftmost top end, we see that $t(p_1) < t(y_2)$. Hence $y_2 \succ p_1$. Since $y_2$ does not intersect $p_1$, it follows that $y_2 \notin L_2$. 
		
		By the definition of $y_2$, we know that $t(y_2) < t(u)$. By Lemma \ref{lem:u-left-of-p2}, it holds that $t(u) < t(p_2)$, thus $t(y_2) < t(p_2)$. Moreover, $b(y_2) > b(p_1) > b(p_2)$ (by Equation \eqref{eq:x-even}). Hence, $y_2$ intersects $p_2$. We conclude that $y_2 \in L_3$. 
		
		Suppose that for all $i < k$, for some $k > 2$, it holds that $y_i \succ p_{i-1}$ and $y_i \in L_{i+1}$.  Now consider $y_k$. Suppose that $k$ is even. 
		By the induction hypothesis and Lemma \ref{lem:neighbourhood}, we know that $y_{k-1} \sim p_{k-1}$, since $y_{k-1} \in L_k$. By definition of $y_k$, it follows that $b(y_k) \geq b(p_{k-1})$. And by Equation \eqref{eq:x-even}, we know that $b(p_{k-1}) > b(p_{k})$, thus $b(y_k) > b(p_{k})$. Similarly, by the definition of $p_k$, we know that $t(p_k) \geq t(y_{k-1})$. And by Equation \eqref{eq:y-even}, we know that $t(y_k) < t(y_{k-1})$, hence $t(p_k) > t(y_k)$. We conclude that $p_k$ intersects $y_k$. It follows that $y_k$ is in layer $k-1$, $k$ or $k+1$. 
		
		Notice that if $y_k \in L_{k-1}$, then the length of $Y_{u,v}$ is at least $d= \diam(G) + 1$. This yields a contradiction with the fact that $Y_{u,v}$ is a shortest $u,v$-path. Thus $y_k \notin L_{k-1}$. 
		Suppose that $y_k \in L_{k}$. Then $y_k$ intersects $p_{k-1}$. We have seen that  $b(y_k) \geq b(p_{k-1})$, thus $t(y_k) < t(p_{k-1})$. By Equation \eqref{eq:x-odd}, we have $b(p_{k-1}) > b(p_{k-2})$, Thus $b(y_k) > b(p_{k-2})$. By Equation \eqref{eq:x-odd}, we also have that $t(p_{k-1}) < t(p_{k-2})$, thus $t(y_k) < t(p_{k-2})$. It follows that $y_k \sim p_{k-2}$. This yields a contradiction with the assumption that $y_k \in L_k$. 
		We conclude that $y_k \in L_{k+1}$.  
		
		The case for $k$ odd is analogous. 
		Since $y_i \in L_{i+1}$ for all internal vertices $y_i$ of $Y_{u,v}$, we conclude that $Y_{u,v}$ is a rainbow path. 
	\end{proof}
	
	\begin{theorem}[=Theorem \ref{thm:permutationIntro}]\label{thm:permutation}
		For every $n$-vertex permutation graph $G$, it holds that $\rvc(G) = \diam(G) - 1$. Moreover, we can compute an optimal rainbow vertex coloring in $O(n^2)$ time.
	\end{theorem}
	\begin{proof}
		By Lemma~\ref{lem:length} we know that either $d = \diam(G)$ or $d = \diam(G) + 1$, and either $d' = \diam(G)$ or
		$d' = \diam(G) + 1$.
		If $d = \diam(G)$ or if $d' = \diam(G)$, we have seen a rainbow coloring of $G$ with $\diam(G) - 1$ colors in Lemma~\ref{lem:layers}. If both $d$ and $d'$ equal $\diam(G) + 1$, then we have seen a coloring of $G$ with $\diam(G) - 1$ colors. Lemmas~\ref{lem:almost-all-paths-rainbow}, \ref{lem:even-more-rainbow-paths} and~\ref{lem:last-rainbow-paths} show that this coloring is indeed a rainbow coloring. We conclude that $\rvc(G) = \diam(G) -1$.
		
		Assume that we are given a permutation model of the graph and thus know the values $t(v)$ and $b(v)$ for each vertex $v \in V(G)$. Otherwise, a permutation model can be computed in linear time~\cite{McConnellS1999}.
		First, compute $d$ and $d'$. Following the description of $P$ and $Q$, this takes linear time. Computing the diameter of $G$ takes $O(n^2)$ time using the algorithm of Mondal et al.~\cite{MondalPP2003}.
		The colorings given by Lemma~\ref{lem:layers} and before Lemma~\ref{lem:almost-all-paths-rainbow} can each be computed in linear time through a breadth-first search. By the preceding arguments, an optimal rainbow vertex coloring can be computed in $O(n^2)$ time.
	\end{proof}

	\section{Split strongly chordal graphs}
In this section, we show that \krvc~and \ksrvc~are polynomial-time solvable on split strongly chordal graphs. We show that this result is tight in the sense that both problems become \NP-complete on split graphs if we forbid any finite family of suns.

In order to prove our next theorem we will use the following property of dually chordal graphs, a graph class that contains that of strongly chordal graphs~\cite{DUALLYCHORDAL}.

\begin{lemma}[{Brandst{\"a}dt et al.~\cite{DUALLYCHORDAL}}]\label{lem:duallychordal}
A graph $G$ is dually chordal if and only if $G$ has a spanning tree $T$ such that every maximal clique of $G$ induces a subtree of $T$.
\end{lemma} 

We prove a tree with a stronger property exists in split strongly chordal graphs.

\begin{lemma} \label{lem:split-spanning}
Let $G=(V,E)$ be a connected split strongly chordal graph, with $V=K\cup S$, where $K$ is a clique and $S$ is an independent set. Then $G$ has a spanning tree $T$ such that every maximal clique of $G$ induces a subtree of $T$ and every vertex of $S$ is a leaf of $T$.
\end{lemma}
\begin{proof}
Since $G$ is a strongly chordal graph (and thus a dually chordal graph~\cite{DUALLYCHORDAL}), by Theorem~\ref{lem:duallychordal}, $G$ has a spanning tree $T'$ such that the vertices of every maximal clique of $G$ induce a subtree of $T'$. We observe the following simple property.

We will now modify $T'$ in order to obtain another spanning tree of $G$ that has the same property and, additionally, is such that every vertex of~$S$ is a leaf.\\ 

\noindent {\bf Tree Modification.} Let $x\in S$ be such that $x$ is not a leaf of $T'$. Let $u,v\in K$ be two neighbors of $x$ in $T'$. Then
\begin{enumerate}
\item[(i)] add the edge $uv$ to $T'$;
\item[(ii)] delete the edge $vx$ from $T'$.
\end{enumerate}

Observe that the result from this operation is still a tree and it is still spanning.

\begin{claim}\label{claim:modissafe}
Let $T''$ be the tree obtained after the application of the Tree Modification to $T'$. Then the following holds:
\begin{enumerate}
\item $d_{T''}(x)<d_{T'}(x)$;
\item Every maximal clique of $G$ induces a subtree of $T''$.
\end{enumerate}
\end{claim}

\begin{claimproof}
It is easy to see that 1.\ holds, since the edge $vx$ was deleted in step (ii) and no other edge incident to $x$ was added. To see that 2.\ also holds note that since $u$ is adjacent to both $v$ and $x$ and since $G[N(x)]$ is a clique, every maximal clique $D$ of $G$ that contains $v$ and $x$ also contains $u$. Since $uv,ux\in E(T'')$, the vertices of $D$ still induce a connected subgraph of $T''$. 
\end{claimproof}

We now iteratively apply Tree modification to $T'$ on non-leaf vertices of $S$. Indeed, if $x \in S$ is not a leaf of $T'$, then it has at least two neighbors in $T'$. 
Since $T'$ is a spanning tree, any neighbors of $x$ in $T'$ must be adjacent to $x$ in $G$, and thus are in $K$. 
Hence, Tree Modification can be applied. By Claim~\ref{claim:modissafe}, we can safely apply Tree Modification repeatedly to $T'$ until we obtain a tree $T$ in which all the vertices of $S$ have degree~$1$ in $T$.
\end{proof}

\begin{theorem}[=Theorem \ref{thm:split-strongly-chordal-Intro}] \label{thm:split-strongly-chordal}
If $G$ is a split strongly chordal graph with $\ell$ cut vertices, then $\rvc(G)=\srvc(G)=\max\{\diam(G)-1,\ell\}$.
\end{theorem}

\begin{proof}
Let $G=(V,E)$ be a split strongly chordal graph, with $V=K\cup S$, where $K$ is a clique and $S$ is an independent set.
Note that if $\diam(G)\leq 2$, we can (strong) rainbow color $G$ by assigning the same color to all the vertices. 
Notice that in this case $G$ has at most one cut vertex, thus $\rvc(G)=\srvc(G)=\max\{\diam(G)-1,\ell\}$.

Assume then that $\diam(G)=3$ (recall that if $G$ is a split graph, then $\diam(G)\leq 3$). 
By Lemma~\ref{lem:split-spanning}, $G$ has a spanning tree $T$ such that every maximal clique of $G$ induces a subtree of $T$ and every vertex of $S$ is a leaf of $T$.
Let $\mathcal{T}$ denote the subtree of $T$ induced by the vertices of $K$, that is, the subtree of $T$ obtained by the deletion of the leaves corresponding to vertices of $S$. Note that $\mathcal{T}$ is a tree with $V_{T} = K$.

\begin{claim}\label{claim:Tisgood}
For every $x\in S$, $N(x)$ induces a subtree of $\mathcal{T}$.
\end{claim}
\begin{claimproof}
Since $N[x]$ is a maximal clique of $G$, the vertices of $N[x]$ induce a subtree of $T$. Since $x$ is a leaf of $T$, we conclude that $N(x)$ indeed induces a subtree of $\mathcal{T}$.
\end{claimproof}

We will now use the tree~$\mathcal{T}$ to provide a (strong) rainbow coloring of $G$.

Suppose first that $G$ is 2-connected, that is, that no vertex of $S$ has degree one in $G$. Color the vertices of $K$ according to a proper 2-coloring of the vertices of $\mathcal{T}$, and give arbitrary colors to the vertices of $S$. Let $\phi$ be the coloring of $G$ obtained in this way. Note that $\phi$ is indeed a (strong) rainbow coloring of $G$. To see this, let $u,v\in V$ be such that $d_G(u,v)= 3$. Since $G$ is a split graph, we have that $u,v\in S$. Since $G$ is 2-connected, $|N(u)|\geq 2$ and $|N(v)|\geq 2$. Moreover, since $N(u)$ and $N(v)$ induce subtrees of $\mathcal{T}$, we know that these two sets are not monochromatic under $\phi$. Thus, there exist $x\in N(u)$ and $y\in N(v)$ such that $\phi(x)\neq \phi(y)$, which shows that $uxyv$ is a rainbow (shortest) path between $u$ and $v$.

We now consider the case in which $G$ has cut vertices. Let $C\subset V$ be the set of cut vertices of $G$. Consider a proper 2-coloring $\phi$ of $\mathcal{T}$. If there exist $c_1,c_2\in C$ such that $\phi(c_1)\neq\phi(c_2)$, then we can obtain a (strong) rainbow coloring for $G$ with $\ell$ colors by assigning distinct colors in the set $\{3,\ldots,\ell\}$ to the remaining cut vertices of $G$. Note that with this coloring of $\mathcal{T}$, it still holds that for every $w\in S$, if $|N(w)|>1$, then $N(w)$ is not monochromatic under $\phi$. Since all the cut vertices were assigned distinct colors, by the same argument used in the 2-connected case, this is indeed a (strong) rainbow coloring of $G$. 
Note that this reasoning also applies if $|C|=1$, so from now on we may assume $|C| \geq 2$.

If all the vertices of $C$ were assigned the same color, since $\phi$ was a proper 2-coloring of $\mathcal{T}$, we have that for every $x,y\in C$, $d_\mathcal{T}(x,y)\geq 2$. Let $c_1,c_2\in C$ be two cut vertices such that the unique path connecting $c_1$ and $c_2$ in $\mathcal{T}$ contains no other vertex of $C$. Let $z$ be the vertex adjacent to $c_1$ in this path. Note that $z\notin C$. We will consider the following coloring $\phi'$ of $\mathcal{T}$. Let $\phi'(c_1)=\phi'(z)=1$. Now we extend $\phi'$ by considering a proper 2-coloring of the subtree of $\mathcal{T}$ rooted in $c_1$ (resp.\ $z$) that assigns color $1$ to the vertex $c_1$ (resp.\ $z$). Note that now we have $\phi'(c_2)=2$. Finally, assign distinct colors from $\{3,\ldots,\ell\}$ to the vertices of $C\setminus\{c_1,c_2\}$. Note that, under this coloring, if there exists $w\in S$ such that $|N(w)|>1$ and $N(w)$ is monochromatic under $\phi'$, then $N(w)=\{c_1,z\}$. To obtain a (strong) rainbow coloring of $G$ with $\ell$ colors, we color the vertices of $K$ according to $\phi'$ and give arbitrary colors to the vertices of $S$. To see that this is indeed a (strong) rainbow coloring of $G$, let $u,v\in V$ be such that $d_G(u,v)= 3$. If $|N(u)|=|N(v)|=1$, the unique shortest path between $u$ and $v$ is a rainbow path, since all cut vertices of $G$ received distinct colors. Assume $|N(u)|\geq 2$. 

Recall that if there exists $w\in S$ such that $|N(w)|\geq 2$ and $N(w)$ is monochromatic under $\phi'$, then $N(w)=\{c_1,z\}$. Moreover, since $d_G(u,v)=3$, we have $N(u)\cap N(v)=\emptyset$. If $|N(v)|\geq 2$, then at most one among $N(u)$ and $N(v)$ is monochromatic. Thus there exists $x\in N(u)$ and $y\in N(v)$ such that $\phi'(x)\neq \phi'(y)$, and hence $uxyv$ is a rainbow (shortest) path between $u$ and $v$. To conclude, consider the case in which $|N(v)|=1$. Recall that $c_1$ is the only cut vertex such that $\phi'(c_1)=1$. Therefore, if $N(v)=\{c_1\}$, then $N(u)$ is \emph{not} monochromatic, which implies that we can find a vertex $x\in N(u)$ such that $\phi'(x)\neq 1$. Finally, if $N(v)=\{c\}$ with $c\neq c_1$, then we can again find $x\in N(u)$ such that $\phi'(x)\neq \phi'(c)$, since $\phi'(c)\neq 1$, $\phi'(c_1)=\phi'(z)=1$ and any monochromatic subtree of $\mathcal{T}$ induced by the open neighborhood of a vertex of $S$ is of the form $\{c_1,z\}$.  
\end{proof}

We now show that both \krvc~and \ksrvc~are \NP-complete if we only forbid a finite number of suns. In what follows, we make use of the same reduction of Heggernes et al.~\cite{MFCS2018} for split graphs. Their reduction is from {\sc Hypergraph Coloring}. Given a hypergraph $\mathcal{H} = (U, \mathcal{E})$, where $U=\{u_1,\ldots, u_n\}$, they construct a split graph $G=(K' \cup I', \; E)$, where $K'=K_1'\cup \dots \cup K_{n+1}'$ with $K_i':=\left\{ u^i_t\mid u_t\in U \right\}$, $I'=I_1'\cup \dots \cup I_{n+1}'$ with $I_i':= \{ x_e^i \mid e \in \mathcal{E} \}$ and $E:=   \{ xy \mid x,y\in K'\}$ $\cup$ $\{ u^i_tx_e^i \mid u_t \in U, e \in \calE, u_t \in e, 1\leq i \leq n+1\}$. 

In our case, we start with an instance of {\sc Graph Coloring} restricted to $(C_3,\ldots,C_p)$-free graphs, a problem that was shown to be \NP-complete by Kr{\' a}l' et al~\cite{KRAL} (see also~\cite{GOLOVACHSURVEY}) for every fixed $k\geq 3$. We can see an input $G=(V,E)$ of {\sc Graph Coloring} as a hypergraph in which every hyperedge has size two. We perform the same construction and obtain a graph $G'$. The fact that $G'$ is a yes instance to \krvc~(and \ksrvc ) if and only if $G$ is a yes instance to {\sc Graph Coloring} follows from \cite[Lemma 11]{MFCS2018}. We now show that $G'$ is a split graph with no induced $t$-sun, with $3\leq t\leq p$.  

\begin{lemma}\label{lem:sunfree}
$G'$ is a split $(S_3,\ldots,S_p)$-free graph.
\end{lemma}
\begin{proof}
Let $G=(V,E)$ be the $(C_3,\ldots,C_p)$-free graph that originated $G'$. It is easy to see that $G'$ is a split graph with clique $K'$ and independent set $I'$. Suppose $G'$ contains an induced $t$-sun, for some $t\in\{3,\ldots,k\}$,  on vertex set $A=\{v_1,\ldots,v_t,w_1,\ldots,w_t\}$, where $G'[\{v_1,\ldots,v_t\}]$ is a clique, $G'[\{w_1,\ldots,w_t\}]$ is an independent set, $w_i$ is adjacent to $v_i$ and $v_{i+1}$ for $1\leq i\leq t-1$ and $w_t$ is adjacent to $v_t$ and $v_1$. Since every vertex of $I'$ has degree two and $K'$ is a clique, we necessarily have that $\{v_1,\ldots,v_t\}\subset K'$ and $\{w_1,\ldots,w_t\}\subset I'$. Moreover, note that, by construction, for every $j$, $N_{G'}(I_j')\subseteq K_j'$. Hence, there exists $p\in \{1,\ldots,n+1\}$ such that $\{v_1,\ldots,v_t\}\subset K_p'$ and $\{w_1,\ldots,w_t\}\subset I_p'$. For simplicity assume $v_i=u^p_i$. Since $w_i$ is adjacent to $u^p_i$ and $u^p_{i+1}$ for $1\leq i<t$, we have that $u_iu_{i+1}\in E(G)$ for $1\leq i<t$. Analogously, it holds that $u_tu_1 \in E(G)$ as well. This implies that $u_1u_2\ldots u_tu_1$ is a cycle in $G$ and therefore $G$ contains an induced cycle of length at most $t$. This is a contradiction since $t\leq p$ and $G$ is $(C_3,\ldots,C_p)$-free.
\end{proof}

From Lemma~\ref{lem:sunfree} and Lemma 11 of~\cite{MFCS2018}, we obtain the following theorem.

\begin{theorem}
For any fixed $p\geq 3$, \krvc~and \ksrvc~are \NP-complete on split $(S_3,\ldots,S_p)$-free graphs for any fixed $p \geq 3$.
\end{theorem}

	\section{Powers of trees}

In this section we study powers of trees. Let $T$ be a tree, and $z$ in the center of $T$. Let $e = zv$ be an edge that is incident to $z$, with $v$ not in the center. When $e$ is removed from the tree, the tree will fall apart in two parts, a \emph{branch} is the part that does not contain $z$. If the center of $T$ contains only one vertex, the number of branches equals the degree of $z$. 

\subsection{Squares of trees}

Two trivial lower bounds for the rainbow coloring number of a graph $G$ are the number of cut vertices in $G$ and $\diam(G) -1$. 
In squares of trees we found graphs that need more than $\diam(T^2)-1$ colors. Notice that squares of trees are 2-connected, so there are no cut vertices. 

\begin{lemma}\label{lem:square-lowerbound}
	Let $T$ be a tree such that the center of $T$ consist of a single vertex $z$, $T$ has diameter at least $6$, and there are at least three branches from the center with maximum length. Then $\srvc(T^2) \geq \rvc(T^2) \geq \diam(T^2)$. 
\end{lemma}
\begin{proof}
	Let $v_1$, $v_2$, and $v_3$ be three vertices with maximum distance to $z$ in three different branches. 
	We consider the case that $\diam(T^2)$ is odd. 
	There is a unique shortest path $P = v_1, p_1, p_2, \ldots, p_k, v_2$ from $v_1$ to $v_2$ in $T^2$.
	Analogously, there is a unique shortest path $Q = v_1, q_1, q_2, \ldots, q_k, v_3$ from $v_1$ to $v_3$ in $T^2$. Notice that $q_1 = p_1$, $q_2 = p_2$, $\ldots$, $q_{j} = p_{j}$, where $j = \lfloor\frac{\diam(T^2)}{2}\rfloor$. That is, $P$ and $Q$ use the same vertices in the branch of $v_1$. The unique shortest path $R$ in $T^2$ from $v_2$ to $v_3$ is $v_2, p_k, \ldots, p_{j+1}, q_{j+1}, \ldots, q_k, v_3$. See Figure \ref{fig:square-diam-lowerbound}. 
	
	We give a proof by contradiction. Let $c$ be a rainbow vertex coloring that uses at most $\diam(T^2) -1$ colors.  
	Notice that the paths $P$, $Q$, and $R$ have length $\diam(T^2)$. 
	Therefore, for each of these paths, all internal vertices are assigned different colors and all colors appear in the path. 
	Since the first $j$ vertices of the paths $P$ and $Q$ are equal, we see that the colors used for $p_{j+1}, \ldots, p_k$ are the same as the colors used for $q_{j+1}, \ldots, q_k$. Since $\diam(T) \geq 6$, $\{p_{j+1}, \ldots, p_k\}$ and $\{q_{j+1}, \ldots, q_k\}$ are non-empty. Hence, there is a color that appears twice in $R$, which yields a contradiction. 
	We conclude that $\rvc(T^2) \geq \diam(T^2)$. 
	
	The case that $\diam(T^2)$ is even is analogous. 
\end{proof}

\begin{figure}
	\centering
	\begin{tikzpicture}
		\node[vertex] (z) at (1.5,0) {};
		\node[vertex, label=left:{$p_2=q_2$}] (a1) at (0,-1) {};
		\node[vertex] (a2) at (0,-2) {};
		\node[vertex, label=left:{$p_1 = q_1$}] (a3) at (0,-3) {};
		\node[vertex] (a4) at (0,-4) {};
		\node[vertex, label=left:$v_1$] (a5) at (0,-5) {};
		\node[vertex, label=left:$p_3$] (b1) at (1.5,-1) {};
		\node[vertex] (b2) at (1.5,-2) {};
		\node[vertex, label=left:$p_4$] (b3) at (1.5,-3) {};
		\node[vertex] (b4) at (1.5,-4) {};
		\node[vertex, label=left:$v_2$] (b5) at (1.5,-5) {};
		\node[vertex, label=left:$q_3$] (c1) at (3,-1) {};
		\node[vertex] (c2) at (3,-2) {};
		\node[vertex, label=left:$q_4$] (c3) at (3,-3) {};
		\node[vertex] (c4) at (3,-4) {};
		\node[vertex, label=left:$v_3$] (c5) at (3,-5) {};
		\draw[edge] (z) -- (a1) -- (a2) -- (a3) -- (a4) -- (a5);
		\draw[edge] (z) -- (b1) -- (b2) -- (b3) -- (b4) -- (b5);
		\draw[edge] (z) -- (c1) -- (c2) -- (c3) -- (c4) -- (c5);
	\end{tikzpicture}
	\caption{A graph $T$ for which $T^2$ needs $\diam(T^2)$ colours, see Lemma \ref{lem:square-lowerbound}.  } \label{fig:square-diam-lowerbound}
\end{figure}
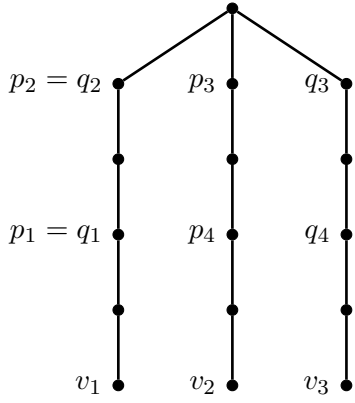

The class of graphs described in the statement of Lemma \ref{lem:square-lowerbound} needs exactly $\diam(T^2)$ colors. 
We define layer $i$ as the set of all vertices with distance $\lfloor\diam(T)/2\rfloor - i$ to the center of $T$. For a vertex $v$, we write $l(v)$ for the layer that it is contained in, so $l(v) = \lfloor\diam(T)/2\rfloor - d$, where $d$ is the distance of $v$ to the center of $T$. 

\begin{lemma} \label{lem:square-upperbound}
	Let $T$ be a tree such that the center of $T$ consist of a single vertex, $T$ has diameter at least $6$, and there are at least three branches from the center with maximum length. Then $\rvc(T^2) = \diam(T^2)$. 
\end{lemma}

\begin{proof}
	Consider the following coloring $c$:
	\begin{align*}
		c(v) = \begin{cases}
			i & \text{if $v$ is in layer $i$, $1\leq i\leq\frac{\diam(T)}{2}$}\\
			1 & \text{otherwise}. 
			\end{cases}
	\end{align*}
	Notice that the number of colors in this coloring is $\frac{\diam(T)}{2} = \diam(T^2)$. 
	
	We claim that this is a rainbow coloring. Let $u$ and $v$ be two vertices in $T^2$. 
	Let $w$ be the lowest common ancestor of $u$ and $v$. Consider the following $u-v$-path: use the even layers to go from $u$ to $w$, and the odd layers to go from $w$ to $v$. See Figure \ref{fig:rainbowpath}.  Notice that the internal vertices of such a path are in different layer, hence this is a rainbow path. 
	
	Combining this with Lemma \ref{lem:square-lowerbound}, we conclude that $\rvc(T^2) = \diam(T^2)$. 
\end{proof}
\begin{figure}
	\centering
	\begin{tikzpicture}
		\node[path, label=left:$6$] (z) at (1.5,0) {};
		\node[vertex, label=left:$5$] (a1) at (0,-1) {};
		\node[path, label=left:$4$] (a2) at (0,-2) {};
		\node[vertex, label=left:$3$] (a3) at (0,-3) {};
		\node[path, label=left:$2$] (a4) at (0,-4) {};
		\node[vertex, label=left:$1$] (a5) at (0,-5) {};
		\node[path, label=right:$u$] (a6) at (0,-6) {};
		\node[path, label=left:$5$] (b1) at (1.5,-1) {};
		\node[vertex, label=left:$4$] (b2) at (1.5,-2) {};
		\node[path, label=left:$3$] (b3) at (1.5,-3) {};
		\node[vertex, label=left:$2$] (b4) at (1.5,-4) {};
		\node[path, label=left:$1$] (b5) at (1.5,-5) {};
		\node[path, label=right:$v$] (b6) at (1.5,-6) {};
		\node[vertex, label=left:$5$] (c1) at (3,-1) {};
		\node[vertex, label=left:$4$] (c2) at (3,-2) {};
		\node[vertex, label=left:$3$] (c3) at (3,-3) {};
		\node[vertex, label=left:$2$] (c4) at (3,-4) {};
		\node[vertex, label=left:$1$] (c5) at (3,-5) {};
		\node[vertex] (c6) at (3,-6) {};
		\draw[edge] (z) -- (a1) -- (a2) -- (a3) -- (a4) -- (a5) -- (a6);
		\draw[edge] (z) -- (b1) -- (b2) -- (b3) -- (b4) -- (b5) -- (b6);
		\draw[edge] (z) -- (c1) -- (c2) -- (c3) -- (c4) -- (c5) -- (c6);
	\end{tikzpicture}
	\caption{A tree $T$. The numbers indicate the colouring of $T^2$. The gray square vertices are a rainbow path in $T^2$ from $u$ to $v$. } \label{fig:rainbowpath}
\end{figure}

Notice that this is not a strong rainbow vertex coloring, since the $u-v$-path described in the proof is not necessarily a shortest path. 

In squares of trees this is the only example that needs more than $\diam(T^2) - 1$ colors. For trees $T$ with $\diam(T) \leq 4$, it holds that $\diam(T^2) \leq 2$. This shows that $\rvc(T^2) = 1$ in those cases. We will distinguishing two cases for the rest of the graphs. 

\begin{lemma}\label{lem:square-two-branches}
	Let $T$ be a tree such that the center of $T$ consist of a single vertex, $T$ has diameter at least $6$, and there are exactly two branches from the center with maximum length. Then $\rvc(T^2) = \diam(T^2)-1$. 
\end{lemma}
\begin{proof}
	Let $B_1$ be one of the branches with maximum length. Let $B_2$ be all other branches, together with the center vertex. 
	Suppose that $\diam(T^2)$ is odd. Consider the following coloring $c$:
	\begin{align*}
	c(v) = \begin{cases}
	i & \text{if $v\in B_2$ is in layer $i$, $1\leq i\leq\diam(T^2)-1$}\\
	2i-1 & \text{if $v\in B_1$ is in layer $2i$, $1\leq i\leq\frac{\diam(T^2)-1}{2}$}\\
	2i & \text{if $v\in B_1$ is in layer $2i-1$, $1\leq i\leq\frac{\diam(T^2)-1}{2}$}\\
	1 & \text{otherwise}. 
	\end{cases}
	\end{align*}
	Intuitively, we color $B_2$ per layer, and we color $B_1$ similar, but with the colors of layer 1 and 2 swapped, and the colors of layers 3 and 4 swapped, etc. See Figure \ref{fig:square-two-branches}. 
	The number of colors used in this coloring equals $\diam(T^2)-1$. 
	
	Let $u$ and $v$ be two vertices of $T$. We claim that there exists a rainbow path between $u$ and $v$. If $u$ and $v$ are both in $B_i$, for $i=1,2$, we can use the same path as in the proof of Lemma \ref{lem:square-upperbound}. Suppose without loss of generality that $l(u) \leq l(v)$ and let $w$ be the lowest common ancestor of $u$ and $v$. Use the even layers to go from $u$ to $w$, and the odd layers to go from $w$ to $v$. Notice that every layer has a unique color, except for the center vertex, which has the same color as layer 1 in $B_2$ and as layer 2 in $B_1$. When this path contains the center vertex, both $u$ and $v$ are in $B_2$. Because this path does not contain a vertex in layer 1 as an internal vertex, this is a rainbow path. 
	
	If $u \in B_i$ and $v\in B_j$, with $i\neq j$, consider the following path. Let $P_1$ be the path from $u$ taking the even layers towards the center, as far as possible, so $P_1$ ends at a neighbour of the center. Let $P_2$ be the same from $v$. Then the path $P_1$ followed by $P_2$ reversed, is a path from $u$ to $v$. This is a rainbow path since the colors used in the even layers of $B_1$ are exactly the colors of the odd layers of $B_2$ and vice versa. We conclude that $c$ is a rainbow coloring. 
	
	Suppose that $\diam(T^2)$ is even. We slightly modify the coloring $c$:
	\begin{align*}
		c(v) = \begin{cases}
			i & \text{if $v\in B_2$ is in layer $i$, $1\leq i\leq\diam(T^2)-1$}\\
			2i+1 & \text{if $v\in B_1$ is in layer $2i$, $1\leq i\leq\frac{\diam(T^2)-1}{2}$}\\
			2i & \text{if $v\in B_1$ is in layer $2i+1$, $1\leq i\leq\frac{\diam(T^2)-1}{2}$}\\
			1 & \text{otherwise}. 
		\end{cases}
	\end{align*}
	The paths constructed above are rainbow paths in this coloring as well. 
\end{proof}
\begin{figure}
	\centering
	\begin{tikzpicture}
	\node[vertex, label=left:$1$] (z) at (1.5,0) {};
	\node[path, label=left:$3$] (a1) at (0,-1) {};
	\node[vertex, label=left:$4$] (a2) at (0,-2) {};
	\node[path, label=left:$1$] (a3) at (0,-3) {};
	\node[vertex, label=left:$2$] (a4) at (0,-4) {};
	\node[vertex, label=left:$1$] (a6) at (0,-5) {};
	\node[vertex, label=left:$2$] (d4) at (-1.5,-4) {};
	\node[path, label=right:$u$] (d6) at (-1.5,-5) {};
	\node[path, label=left:$4$] (b1) at (1.5,-1) {};
	\node[vertex, label=left:$3$] (b2) at (1.5,-2) {};
	\node[path, label=left:$2$] (b3) at (1.5,-3) {};
	\node[vertex, label=left:$1$] (b4) at (1.5,-4) {};
	\node[path, label=right:$v$] (b6) at (1.5,-5) {};
	\node[vertex, label=left:$3$] (e2) at (3,-2) {};
	\node[vertex, label=left:$2$] (e3) at (3,-3) {};
	\node[vertex, label=left:$1$] (e4) at (3,-4) {};
	\node[vertex, label=left:$4$] (c1) at (4.5,-1) {};
	\node[vertex, label=left:$3$] (c2) at (4.5,-2) {};
	\node[vertex, label=left:$2$] (c3) at (4.5,-3) {};
	\node[vertex, label=left:$1$] (c4) at (4.5,-4) {};
	\draw[edge] (z) -- (a1) -- (a2) -- (a3) -- (a4) -- (a6);
	\draw[edge] (a3) -- (d4) -- (d6);
	\draw[edge] (z) -- (b1) -- (b2) -- (b3) -- (b4) -- (b6);
	\draw[edge] (b1) -- (e2) -- (e3) -- (e4);
	\draw[edge] (z) -- (c1) -- (c2) -- (c3) -- (c4);
	\end{tikzpicture}\qquad\qquad
	\begin{tikzpicture}
		\node[path, label=left:$1$] (z) at (1.5,0) {};
		\node[vertex, label=left:$4$] (a1) at (0,-1) {};
		\node[path, label=left:$5$] (a2) at (0,-2) {};
		\node[vertex, label=left:$2$] (a3) at (0,-3) {};
		\node[vertex, label=left:$3$] (a4) at (0,-4) {};
		\node[vertex, label=left:$1$] (a5) at (0,-5) {};
		\node[vertex, label=left:$1$] (a6) at (0,-6) {};
		\node[path, label=left:$3$] (d4) at (-1.5,-4) {};
		\node[vertex, label=left:$1$] (d5) at (-1.5,-5) {};
		\node[path, label=right:$u$] (d6) at (-1.5,-6) {};
		\node[vertex, label=left:$5$] (b1) at (1.5,-1) {};
		\node[path, label=left:$4$] (b2) at (1.5,-2) {};
		\node[vertex, label=left:$3$] (b3) at (1.5,-3) {};
		\node[path, label=left:$2$] (b4) at (1.5,-4) {};
		\node[vertex, label=left:$1$] (b5) at (1.5,-5) {};
		\node[path, label=right:$v$] (b6) at (1.5,-6) {};
		\node[vertex, label=left:$4$] (e2) at (3,-2) {};
		\node[vertex, label=left:$3$] (e3) at (3,-3) {};
		\node[vertex, label=left:$2$] (e4) at (3,-4) {};
		\node[vertex, label=left:$1$] (e5) at (3,-5) {};
		\node[vertex, label=left:$5$] (c1) at (4.5,-1) {};
		\node[vertex, label=left:$4$] (c2) at (4.5,-2) {};
		\node[vertex, label=left:$3$] (c3) at (4.5,-3) {};
		\node[vertex, label=left:$2$] (c4) at (4.5,-4) {};
		\node[vertex, label=left:$1$] (c5) at (4.5,-5) {};
		\draw[edge] (z) -- (a1) -- (a2) -- (a3) -- (a4) -- (a5) -- (a6);
		\draw[edge] (a3) -- (d4) -- (d5) -- (d6);
		\draw[edge] (z) -- (b1) -- (b2) -- (b3) -- (b4) -- (b5) -- (b6);
		\draw[edge] (b1) -- (e2) -- (e3) -- (e4) -- (e5);
		\draw[edge] (z) -- (c1) -- (c2) -- (c3) -- (c4) -- (c5);
	\end{tikzpicture}
	\caption{A tree $T$. The numbers indicate a rainbow colouring of $T^2$. The gray square vertices are a rainbow path from $u$ to $v$. See Lemma \ref{lem:square-two-branches}.  } \label{fig:square-two-branches}
\end{figure}

Again, this is not a strong rainbow coloring.

\begin{lemma}\label{lem:square-center-two-vertices}
	Let $T$ be a tree such that the center of $T$ consist of two vertices and $T$ has diameter at least $5$. Then $\rvc(T^2) = \diam(T^2)-1$. 
\end{lemma}
\begin{proof}
	We will color per layer: 
	\begin{align*}
	c(v) = \begin{cases}
	i & \text{if $v$ is in layer $i$, $1\leq i\leq\frac{\diam(T)}{2}$}\\
	1 & \text{otherwise}. 
	\end{cases}
	\end{align*}
	
	The number of colors used in this coloring equals $\lfloor\frac{\diam(T)}{2}\rfloor = \diam(T^2)-1$. 
	
	We claim that this is a rainbow coloring. 
	Let $u$ and $v$ be two vertices of $T$. Write $z_1$ and $z_2$ for the two center vertices. If $u$ and $v$ are both in a branch of $z_i$, for $i=1,2$, we can use the same path as in the proof of Lemma \ref{lem:square-upperbound}: Let $w$ be the lowest common ancestor of $u$ and $v$. Use the even layers to go from $u$ to $w$, and the odd layers to go from $w$ to $v$. Since we use at most one vertex per layer, this is a rainbow path. 
	
	Otherwise, assume without loss of generality that $u$ is in a branch of $z_1$ and $v$ in a branch of $z_2$. Use the even layers to go from $u$ to $z_1$, and the odd layers to go from $z_1$ to $v$. Again, this is a rainbow path. 
	See Figure \ref{fig:square-center-two-vertices}. 
\end{proof}
\begin{figure}
	\centering
	\begin{tikzpicture}
	\node[path, label=left:$6$] (z1) at (.75,0) {};
	\node[vertex, label=right:$6$] (z2) at (2.25,0) {};
	\node[vertex, label=left:$5$] (a1) at (0,-1) {};
	\node[path, label=left:$4$] (a2) at (0,-2) {};
	\node[vertex, label=left:$3$] (a3) at (0,-3) {};
	\node[path, label=left:$2$] (a4) at (0,-4) {};
	\node[vertex, label=left:$1$] (a5) at (0,-5) {};
	\node[path, label=right:$u$] (a6) at (0,-6) {};
	\node[path, label=left:$5$] (b1) at (1.5,-1) {};
	\node[vertex, label=left:$4$] (b2) at (1.5,-2) {};
	\node[path, label=left:$3$] (b3) at (1.5,-3) {};
	\node[vertex, label=left:$2$] (b4) at (1.5,-4) {};
	\node[path, label=left:$1$] (b5) at (1.5,-5) {};
	\node[path, label=right:$v$] (b6) at (1.5,-6) {};
	\node[vertex, label=left:$5$] (c1) at (3,-1) {};
	\node[vertex, label=left:$4$] (c2) at (3,-2) {};
	\node[vertex, label=left:$3$] (c3) at (3,-3) {};
	\node[vertex, label=left:$2$] (c4) at (3,-4) {};
	\node[vertex, label=left:$1$] (c5) at (3,-5) {};
	\node[vertex] (c6) at (3,-6) {};
	\draw[edge] (z1) -- (a1) -- (a2) -- (a3) -- (a4) -- (a5) -- (a6);
	\draw[edge] (z2) -- (b1) -- (b2) -- (b3) -- (b4) -- (b5) -- (b6);
	\draw[edge] (z2) -- (c1) -- (c2) -- (c3) -- (c4) -- (c5) -- (c6);
	\draw[edge] (z1) -- (z2);
	\end{tikzpicture}
	\caption{A tree $T$. The numbers indicate the colouring of $T^2$. The gray square vertices are a rainbow path in $T^2$ from $u$ to $v$. } \label{fig:square-center-two-vertices}
\end{figure}
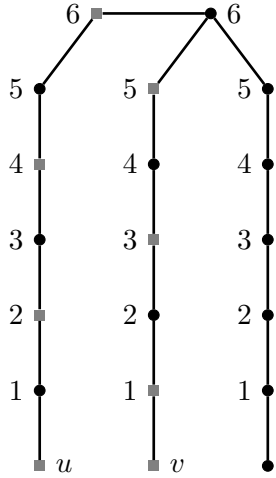

\subsection{Higher powers of trees}

In this section we consider higher powers $T^k$ for $k\geq3$. We will see that the same results hold for the higher powers, but we have to distinguish more cases. Specifically, for the cases where the center consists of a single vertex, we have to consider whether $\diam(T) \equiv 0 \pmod k$ or not. For squares of trees we did not consider the case that $\diam(T) \not\equiv 0 \pmod k$, since the diameter is always even if the center consists of a single vertex. 

\begin{lemma}\label{lem:higher-lowerbound}
	Let $T$ be a tree such that the center of $T$ consist of a single vertex, $T$ has diameter at least $3k$, $\diam(T) \equiv 0 \pmod k$, and there are at least three branches from the center with maximum length. Then $\srvc(T^k) \geq \rvc(T^k) \geq \diam(T^k)$. 
\end{lemma}
\begin{proof}
	This proof is analogous to the proof of Lemma \ref{lem:square-lowerbound}. 
\end{proof}

As for squares of trees, we can proof that $\diam(T^k)$ colors suffices in this case.
\begin{lemma} \label{lem:higher-upperbound}
	Let $T$ be a tree such that the center of $T$ consist of a single vertex, $T$ has diameter at least $3k$, $\diam(T) \equiv 0 \pmod k$, and there are at least three branches from the center with maximum length. Then $\rvc(T^k) = \diam(T^k)$. 
\end{lemma}

\begin{proof}
	To describe a rainbow coloring with $\diam(T^k)$ colors, we need to count the layers bottom up. That is, we start counting the layers at the vertices that are furthest from the center. Recall that $l(v)$ is the number of the layer that contains $v$. 
	
	Consider the following coloring $c$:
	\begin{align*}
		c(v) = \begin{cases}
			1 & \text{if $v$ is the center vertex,}\\
			i & \text{if $l(v)=i$, and $i\equiv 0 \pmod k$ or $i\equiv -1 \pmod k$, $0 < i < \frac{\diam(T)}{2}$,}\\
			1 & \text{otherwise}. 
		\end{cases}
	\end{align*}
	We will count the number of colors we used. We have one color for every layer $i$ with $i \equiv -1$, or $i \equiv 0 \pmod k$ minus layer 0 and the layer with the center vertex, and one extra color $1$ for the rest of the vertices. There are $\frac{\diam(T)}{2}$ layers excluding the center vertex. Divide the layers $0, 1, 2, \ldots, \frac{\diam(T)}{2} -1$ into blocks of size $k$ (the topmost block does not need to be a complete block). The vertices in layers $-1\pmod k$ are exactly the topmost vertices of the complete blocks, and the vertices in layers $0\pmod k$ are exactly the lowest vertices in the blocks. The number of complete blocks is $\lfloor\frac{\diam(T)/2}{k}\rfloor$, thus the number of layers that is $-1 \pmod k$, is $\lfloor\frac{\diam(T)/2}{k}\rfloor$. And the total number of blocks is $\lceil\frac{\diam(T)/2}{k}\rceil$, so the number of layers that is $0 \pmod k$, minus layer 0, is $\lceil\frac{\diam(T)/2}{k}\rceil - 1$. So the total number of colors used in this coloring is \begin{align*}
	\left\lfloor\frac{\diam(T)/2}{k}\right\rfloor + \left\lceil\frac{\diam(T)/2}{k}\right\rceil - 1 + 1 &=  \left\lfloor\frac{\diam(T^k)}{2}\right\rfloor + \left\lceil\frac{\diam(T^k)}{2}\right\rceil = \diam(T^k).
	\end{align*} 
	
	We claim that this is a rainbow coloring. Let $u$ and $v$ be two vertices in $T^k$. 
	Let $w$ be the lowest common ancestor of $u$ and $v$. Consider the following $u-v$-path: use the layers that are $0 \pmod k$ to go from $u$ to $w$, and the layers that are $-1\pmod k$ to go from $w$ to $v$. 
	See Figure \ref{fig:higher-upperbound}. 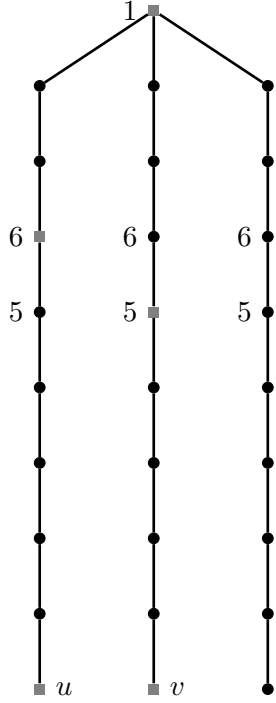
\begin{figure}
	\centering
	\begin{tikzpicture}
		\node[path, label=left:$1$] (z) at (1.5,0) {};
		\node[vertex] (a1) at (0,-1) {};
		\node[vertex] (a2) at (0,-2) {};
		\node[path, label=left:$6$] (a3) at (0,-3) {};
		\node[vertex, label=left:$5$] (a4) at (0,-4) {};
		\node[vertex] (a5) at (0,-5) {};
		\node[vertex] (a6) at (0,-6) {};
		\node[vertex] (a7) at (0,-7) {};
		\node[vertex] (a8) at (0,-8) {};
		\node[path, label=right:$u$] (a9) at (0,-9) {};
		\node[vertex] (b1) at (1.5,-1) {};
		\node[vertex] (b2) at (1.5,-2) {};
		\node[vertex, label=left:$6$] (b3) at (1.5,-3) {};
		\node[path, label=left:$5$] (b4) at (1.5,-4) {};
		\node[vertex] (b5) at (1.5,-5) {};
		\node[vertex] (b6) at (1.5,-6) {};
		\node[vertex] (b7) at (1.5,-7) {};
		\node[vertex] (b8) at (1.5,-8) {};
		\node[path, label=right:$v$] (b9) at (1.5,-9) {};
		\node[vertex] (c1) at (3,-1) {};
		\node[vertex] (c2) at (3,-2) {};
		\node[vertex, label=left:$6$] (c3) at (3,-3) {};
		\node[vertex, label=left:$5$] (c4) at (3,-4) {};
		\node[vertex] (c5) at (3,-5) {};
		\node[vertex] (c6) at (3,-6) {};
		\node[vertex] (c7) at (3,-7) {};
		\node[vertex] (c8) at (3,-8) {};
		\node[vertex] (c9) at (3,-9) {};
		\draw[edge] (z) -- (a1) -- (a2) -- (a3) -- (a4) -- (a5) -- (a6) -- (a7) -- (a8) -- (a9);
		\draw[edge] (z) -- (b1) -- (b2) -- (b3) -- (b4) -- (b5) -- (b6) -- (b7) -- (b8) -- (b9);
		\draw[edge] (z) -- (c1) -- (c2) -- (c3) -- (c4) -- (c5) -- (c6) -- (c7) -- (c8) -- (c9);
	\end{tikzpicture}
	\caption{A tree $T$. Let $k=6$. The numbers indicate the coloring of $T^6$, all unlabeled vertices have color $1$.  The gray square vertices are a rainbow path in $T^6$ from $u$ to $v$. See Lemma \ref{lem:higher-upperbound}.} \label{fig:higher-upperbound}
\end{figure}
	Notice that the internal vertices of such a path are in different layer, all of which are $0$ or $-1 \pmod k$, except for vertex $w$. Hence this is a rainbow path. 
	
	Combining this with Lemma \ref{lem:higher-lowerbound}, we conclude that $\rvc(T^k) = \diam(T^k)$. 
\end{proof}

As in squares of trees, these are the only powers of trees with $\rvc(T^k) \geq \diam(T^k)$. For all other powers of trees it holds that $\rvc(T^k) = \diam(T^k)-1$. 
If $\diam(T) \leq 2k$, then $\diam(T^k) \leq 2$, so then one color suffices. 

\begin{lemma} \label{lem:higher-three-branches}
	Let $T$ be a tree such that the center of $T$ consist of a single vertex, $T$ has diameter at least $2k+1$, and $\diam(T) \not\equiv 0 \pmod k$. 
	Then $\rvc(T^k) = \diam(T^k)-1$. 
\end{lemma}
\begin{proof}
	We will start with a partial coloring of the vertices,
	almost the same coloring as in the proof of Lemma \ref{lem:higher-upperbound}:
	\begin{align*}
	c(v) = \begin{cases}
	i & \text{if $l(v)=i$, and $i\equiv 0 \pmod k$ or $i\equiv -1 \pmod k$, $0 < i < \frac{\diam(T)}{2}$}. 
	\end{cases}
	\end{align*}

	We colored every layer that is $0$ or $-1 \pmod k$, excluding layer $0$ and the center vertex, with a unique color, see Figure \ref{fig:higher-three-branches}. \begin{figure}
	\centering
	\begin{tikzpicture}
	\node[vertex] (z) at (1.5,-2) {};
	\node[vertex] (a3) at (0,-3) {};
	\node[vertex] (a4) at (0,-4) {};
	\node[vertex, label=left:$5$] (a5) at (0,-5) {};
	\node[vertex, label=left:$4$] (a6) at (0,-6) {};
	\node[vertex] (a7) at (0,-7) {};
	\node[vertex] (a8) at (0,-8) {};
	\node[vertex] (a9) at (0,-9) {};
	\node[vertex] (a10) at (0,-10) {};
	\node[vertex] (b3) at (1.5,-3) {};
	\node[vertex] (b4) at (1.5,-4) {};
	\node[vertex, label=left:$5$] (b5) at (1.5,-5) {};
	\node[vertex, label=left:$4$] (b6) at (1.5,-6) {};
	\node[vertex] (b7) at (1.5,-7) {};
	\node[vertex] (b8) at (1.5,-8) {};
	\node[vertex] (b9) at (1.5,-9) {};
	\node[vertex] (b10) at (1.5,-10) {};
	\node[vertex] (c3) at (3,-3) {};
	\node[vertex] (c4) at (3,-4) {};
	\node[vertex, label=left:$5$] (c5) at (3,-5) {};
	\node[vertex, label=left:$4$] (c6) at (3,-6) {};
	\node[vertex] (c7) at (3,-7) {};
	\node[vertex] (c8) at (3,-8) {};
	\node[vertex] (c9) at (3,-9) {};
	\node[vertex] (c10) at (3,-10) {};
	\draw[edge] (z) -- (a3) -- (a4) -- (a5) -- (a6) -- (a7) -- (a8) -- (a9) -- (a10);
	\draw[edge] (z) -- (b3) -- (b4) -- (b5) -- (b6) -- (b7) -- (b8) -- (b9) -- (b10);
	\draw[edge] (z)-- (c3) -- (c4) -- (c5) -- (c6) -- (c7) -- (c8) -- (c9) -- (c10);
	\end{tikzpicture}
	\qquad\qquad
	\begin{tikzpicture}
	\node[vertex] (z) at (1.5,-3) {};
	\node[vertex] (a4) at (0,-4) {};
	\node[vertex, label=left:$5$] (a5) at (0,-5) {};
	\node[vertex, label=left:$4$] (a6) at (0,-6) {};
	\node[vertex] (a7) at (0,-7) {};
	\node[vertex] (a8) at (0,-8) {};
	\node[vertex] (a9) at (0,-9) {};
	\node[vertex] (a10) at (0,-10) {};
	\node[vertex] (b4) at (1.5,-4) {};
	\node[vertex, label=left:$5$] (b5) at (1.5,-5) {};
	\node[vertex, label=left:$4$] (b6) at (1.5,-6) {};
	\node[vertex] (b7) at (1.5,-7) {};
	\node[vertex] (b8) at (1.5,-8) {};
	\node[vertex] (b9) at (1.5,-9) {};
	\node[vertex] (b10) at (1.5,-10) {};
	\node[vertex] (c4) at (3,-4) {};
	\node[vertex, label=left:$5$] (c5) at (3,-5) {};
	\node[vertex, label=left:$4$] (c6) at (3,-6) {};
	\node[vertex] (c7) at (3,-7) {};
	\node[vertex] (c8) at (3,-8) {};
	\node[vertex] (c9) at (3,-9) {};
	\node[vertex] (c10) at (3,-10) {};
	\draw[edge] (z) -- (a4) -- (a5) -- (a6) -- (a7) -- (a8) -- (a9) -- (a10);
	\draw[edge] (z) -- (b4) -- (b5) -- (b6) -- (b7) -- (b8) -- (b9) -- (b10);
	\draw[edge] (z) -- (c4) -- (c5) -- (c6) -- (c7) -- (c8) -- (c9) -- (c10);
	\end{tikzpicture}
	\caption{Trees $T$, $k=5$. The numbers indicate the partial coloring of $T^k$, as in Lemma \ref{lem:higher-three-branches}.  
	} \label{fig:higher-three-branches}
\end{figure}
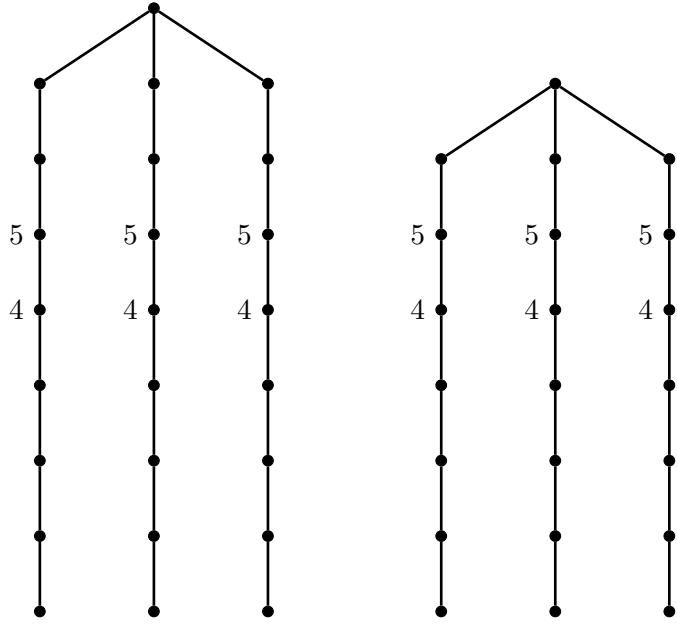
	The idea of a rainbow path from a vertex $u$ to a vertex $v$ will be, as in Lemma \ref{lem:higher-upperbound}, that we walk from $u$ to the lowest common ancestor $w$ using the layer $0 \pmod k$ and then from $w$ to $v$ via the layer $-1\pmod k$. This time, since we want to use only $\diam(T^k) - 1$ colors, we have to be careful about the color that we assign to the lowest common ancestor $w$. 
	
	First, we count the number of colors we used so far. The number of layers $l$ excluding $0$ and the center is $l=\frac{\diam(T)}{2} - 1$. We will divide those layers in blocks of size $k$, starting at layer $1$ (the topmost block does not need to be a complete block). Notice that the top two layers in every complete block are colored. There are $\lfloor \frac{l}{k} \rfloor$ complete blocks, so $2 \lfloor \frac{l}{k} \rfloor$ colors for those blocks. If the topmost block is not a complete block, it uses a color only if it has size $k-1$. But then the center vertex is in a layer $0 \pmod k$, which implies that $\diam(T) \equiv 0 \pmod k$, a contradiction with the assumptions. We conclude that the number of colors used is $2 \lfloor \frac{l}{k} \rfloor$. 	
	
	We distinguish cases for the coloring of rest of the vertices. 
	Let $z$ be the center vertex, and let $B_1$ and $B_2$ be two longest branches from $z$. Define $b_1$ as a vertex in $B_1$ in the highest layer $0 \pmod k$ that is colored and $b_2$ as a vertex in $B_2$ in the highest layer $0 \pmod k$ that is colored. Let $d(b_1, b_2)$ be the distance between $b_1$ and $b_2$ in $T$. 
	
	Suppose that $d(b_1, b_2) > k$. We claim that $\diam(T^k)-1 \geq 2\lfloor\frac{l}{k}\rfloor + 1$. Let $u \in B_1$, $v\in B_2$ be vertices in layer $0$. A $u,v$-path contains a vertex in every complete block in $B_1$, a vertex in every complete block in $B_2$ and a vertex in a topmost incomplete block or $z$. All in all, this are $2\lfloor\frac{l}{k}\rfloor + 1$ internal vertices. So $\diam(T^k) - 1 \geq 2\lfloor\frac{l}{k}\rfloor + 1$. It follows that we can use one more color in our coloring $c$. Use this extra color for all uncolored vertices. Then, for any two vertices $u, v$, the path described before is a rainbow path: use the $0\pmod k$ layers to go from $u$ to $z$ and the $-1\pmod k$ layers to go from $z$ to $v$.  
	
	Now suppose that $d(b_1, b_2) = k$. It follows that $k \mid \diam(T)$, a contradiction with the assumptions of the lemma.  
	
	Now suppose that $d(b_1, b_2) \leq k-1$. Color all uncolored vertices with the highest color used, that is color $ki$ where $i = \lfloor\frac{\diam(T)/2-1}{k}\rfloor$. Let $u$ and $v$ be two vertices. We distinguish two cases. 
	
	If the lowest common ancestor $w$ is in layer $ki$ or lower, we use the same path described above: use the $0\pmod k$ layers to go from $u$ to $w$ and the $-1\pmod k$ layers to go from $w$ to $v$. Because $w$ is in layer $ki$ or lower, no vertex of layer $ki$ is used as internal vertex of the $u-w$ and $w-v$ paths. So the color of $w$ is unique in the path, and this is a rainbow path. 
	
	If $w$ is above layer $ki$, then we use the same path but exclude vertex $w$. Write $w_1$ for the last internal vertex of the $u-w$-path, and $w_2$ for the first internal vertex of the $w-v$-path. Notice that $w_1$ is in layer $ki$ and $w_2$ is in layer $ki-1$. As $d(b_1, b_2) \leq k-1$, the distance between $w_1$ and $w_2$ is at most $k$, so there exists an edge $w_1w_2$ in $T^k$. We conclude that the path using the $0 \pmod k$ vertices to go from $u$ to $w_1$ combined with the path using the $-1 \pmod k$ vertices to go from $w_2$ to $v$ is a rainbow path in $T^k$.  
\end{proof}

\begin{lemma}\label{lem:higher-two-branches}
	Let $T$ be a tree such that the center of $T$ consist of a single vertex, $T$ has diameter at least $3k$, $\diam(T) \equiv 0 \pmod k$, and there are exactly two branches from the center with maximum length. Then $\rvc(T^k) = \diam(T^k)-1$. 
\end{lemma}
\begin{proof}
	To prove this, we will combine the ideas of Lemma \ref{lem:square-two-branches} and \ref{lem:higher-three-branches}. 
	Let $B_1$ and $B_2$ be the branches with maximum length. Let $B_3$ be all other branches. 
	Suppose that $\diam(T^k)$ is even. Consider the following coloring $c$, see Figure \ref{fig:higher-two-branches-even}:
	\begin{align*}
	c(v) = \begin{cases}
	i & \text{if $v\in B_2$ is in layer $i$, $i \equiv 0, -1 \pmod k$, $1\leq i<\diam(T)/2-1$}\\
	ki-1 & \text{if $v\in B_1$ is in layer $ki$, $1\leq ki < \diam(T)/2-1$}\\
	ki & \text{if $v\in B_1$ is in layer $ki-1$, $1\leq ki - 1 < \diam(T)/2-1$}\\
	ki & \text{if $v\in B_3$ is in layer $ki$, $1\leq ki<\diam(T)/2-1$}\\
	ki-1 & \text{if $v\in B_3$ is in layer $ki+1$, $1\leq ki+1 <\diam(T)/2-1$}\\
	1 & \text{otherwise}. 
	\end{cases}
	\end{align*}

	First, we count the number of colors we used. We used the colors $ki$ and $ki - 1$ for $1\leq ki - 1 < \diam(T)/2-1$. These are \begin{align*}
		2\left\lfloor \frac{\diam(T)/2-1}{k}\right\rfloor 
		&= 2\left\lfloor \frac{\diam(T)/k}{2} - \frac{1}{k}\right\rfloor \\
		&= 2 \left\lfloor \frac{\diam(T^k)}{2} - \frac{1}{k}\right\rfloor\\
		&= 2 \left( \frac{\diam(T^k)}{2} - 1 \right)\\
		&= \diam(T^k) - 2 
		\end{align*} 
	colors. Here, the third line holds since we assumed that $\diam(T^k)$ is even. We used one extra color for the rest of the vertices, which makes a total of $\diam(T^k) - 1$ colors. 
	
	Let $u$ and $v$ be two vertices of $T$. We claim that there exists a rainbow path between $u$ and $v$ in $T^k$. We will distinguish several cases. 
	
	Suppose that $u \in B_1$, $v\notin B_1$. Let $z$ be the center of the graph. Use the layers $0 \pmod k$, to walk from $u$ to $z$, and the layers $0\pmod k$, to go from $z$ to $v$. This is a rainbow path. 
	
	Suppose that $u \in B_i$, $v\in B_i$, $i= 1,2$. Let $w$ be the lowest common ancestor of $u$ and $v$. Use the layers $0\pmod k$ to go from $u$ to $w$, and the layers $-1 \pmod k$ to go from $w$ to $v$. This is a rainbow path. 
	
	Suppose that $u \in B_i$, $v\in B_3$, $i = 2,3$. Let $w$ be the lowest common ancestor of $u$ and $v$. Use the layers $0 \pmod k$ to walk from $u$ to $w$, and the layers $1\pmod k$ to go from $w$ to $v$. This is a rainbow path.  
	
	Now suppose that $\diam(T^k)$ is odd. We slightly change the coloring: use color $ki$ with $i = \lfloor\frac{\diam(T)/2-1}{k}\rfloor$ for the case `otherwise' instead of color $1$. Notice that this is the highest color used in the rest of the coloring.   
	The number of colors used is: \begin{align*}
		2\left\lfloor\frac{\diam(T)/2-1}{k}\right\rfloor 
		&= 2\left\lfloor\frac{\diam(T)/k}{2} - \frac{1}{k}\right\rfloor \\
		&= 2\left\lfloor\frac{\diam(T^k)}{2} - \frac{1}{k}\right\rfloor \\
		&\leq 2  \left\lfloor\frac{\diam(T^k)}{2}\right\rfloor \\
		&= \diam(T^k) - 1
		\end{align*}

We can use almost the same paths as in the case of $\diam(T^k)$ being even. 

Suppose that $u \in B_1$, $v\notin B_1$. Let $z$ be the center of the graph. Let $P$ be the path from $u$ to $z$ that uses the layers $0 \pmod k$, and $Q$ the path from $z$ to $v$ that uses the layers $0\pmod k$. Let $p_j$ be the last vertex before $z$ in $P$ and let $q_1$ be the first vertex after $z$ in $Q$. Notice that $z$ has the same color as $q_1$, so we cannot include both $z$ and $q_1$ in a rainbow path. Notice that both $p_j$ and $q_1$ are in the highest layer $ki$. Since the diameter $\diam(T) \equiv 0 \pmod k$ and $\diam(T^k)$ is odd, it follows that the distance between $p_j$ and $q_1$ in $T$ is exactly $k$. Thus there is an edge $p_jq_1$ in $T^k$. Combine the paths $P$ and $Q$ but exclude vertex $z$, to obtain a path from $u$ to $v$. This is a rainbow path. 

Consider the other two cases, so, suppose that $u \in B_i$, $v\in B_i$, $i= 1,2$ or $u \in B_i$, $v\in B_3$, $i = 2,3$. Let $w$ be the lowest common ancestor of $u$ and $v$. If $w$ is in layer $i$ with $i \leq k\lfloor\frac{\diam(T)/2-1}{k}\rfloor$, then use the same path as in the case $\diam(T^k)$ even. If $w$ is in layer $i$ with $i > k\lfloor\frac{\diam(T)/2-1}{k}\rfloor$, then $w$ has the same color as its predecessor $p_j$ or successor $q_1$ in the path described in the case $\diam(T^k)$ even. As in the previous case, the distance between $p_j$ and $q_1$ in $T$ is at most $k$, so there exists an edge $p_jq_1$ in $T^k$. So if we exclude $w$ from the path described in the case $\diam(T^k)$ even, we still have a path, and this is a rainbow path. 
\end{proof}
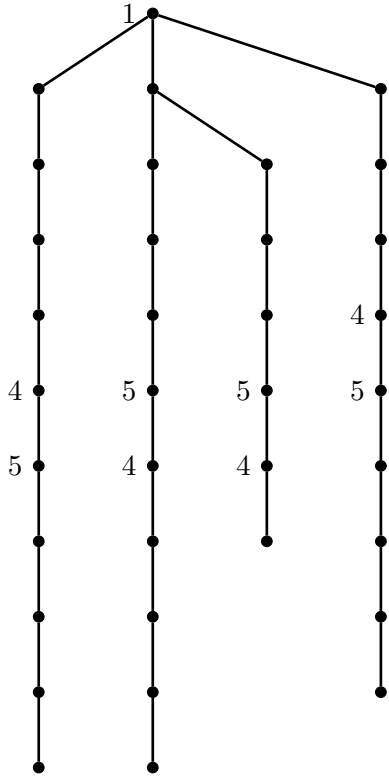
\begin{figure}
	\centering
	\begin{tikzpicture}
	\node[vertex, label=left:$1$] (z) at (1.5,0) {};
	\node[vertex] (a1) at (0,-1) {};
	\node[vertex] (a2) at (0,-2) {};
	\node[vertex] (a3) at (0,-3) {};
	\node[vertex] (a4) at (0,-4) {};
	\node[vertex, label=left:$4$] (a5) at (0,-5) {};
	\node[vertex, label=left:$5$] (a6) at (0,-6) {};
	\node[vertex] (a7) at (0,-7) {};
	\node[vertex] (a8) at (0,-8) {};
	\node[vertex] (a9) at (0,-9) {};
	\node[vertex] (a10) at (0,-10) {};
	\node[vertex] (b1) at (1.5,-1) {};
	\node[vertex] (b2) at (1.5,-2) {};
	\node[vertex] (b3) at (1.5,-3) {};
	\node[vertex] (b4) at (1.5,-4) {};
	\node[vertex, label=left:$5$] (b5) at (1.5,-5) {};
	\node[vertex, label=left:$4$] (b6) at (1.5,-6) {};
	\node[vertex] (b7) at (1.5,-7) {};
	\node[vertex] (b8) at (1.5,-8) {};
	\node[vertex] (b9) at (1.5,-9) {};
	\node[vertex] (b10) at (1.5,-10) {};
	\node[vertex] (e2) at (3,-2) {};
	\node[vertex] (e3) at (3,-3) {};
	\node[vertex] (e4) at (3,-4) {};
	\node[vertex, label=left:$5$] (e5) at (3,-5) {};
	\node[vertex, label=left:$4$] (e6) at (3,-6) {};
	\node[vertex] (e7) at (3,-7) {};
	\node[vertex] (c1) at (4.5,-1) {};
	\node[vertex] (c2) at (4.5,-2) {};
	\node[vertex] (c3) at (4.5,-3) {};
	\node[vertex, label=left:$4$] (c4) at (4.5,-4) {};
	\node[vertex, label=left:$5$] (c5) at (4.5,-5) {};
	\node[vertex] (c6) at (4.5,-6) {};
	\node[vertex] (c7) at (4.5,-7) {};
	\node[vertex] (c8) at (4.5,-8) {};
	\node[vertex] (c9) at (4.5,-9) {};
	\draw[edge] (z) -- (a1) -- (a2) -- (a3) -- (a4) -- (a5) -- (a6) -- (a7) -- (a8) -- (a9) -- (a10);
	\draw[edge] (z) -- (b1) -- (b2) -- (b3) -- (b4) -- (b5) -- (b6) -- (b7) -- (b8) -- (b9) -- (b10);
	\draw[edge] (b1) -- (e2) -- (e3) -- (e4) -- (e5) -- (e6) -- (e7);
	\draw[edge] (z) -- (c1) -- (c2) -- (c3) -- (c4) -- (c5) -- (c6) -- (c7) -- (c8) -- (c9);
	\end{tikzpicture}
	\caption{A tree $T$. The numbers indicate a rainbow colouring of $T^k$, for $k=5$, the unlabeled vertices have color $1$. See Lemma \ref{lem:higher-two-branches}.  } \label{fig:higher-two-branches-even}
\end{figure}

\begin{lemma}\label{lem:higher-center-two-vertices}
	Let $T$ be a tree such that the center of $T$ consist of two vertices and $T$ has diameter at least $2k+1$. Then $\rvc(T^k) = \diam(T^k)-1$. 
\end{lemma}
\begin{proof}
	Let $z_1$ and $z_2$ be the two center vertices. Let $B_1$ be all branches of $z_1$, including $z_1$, and $B_2$ be all branches of $z_2$, including $z_2$. 
	
	Divide the layers $1, 2, \ldots, \frac{\diam(T) - 1}{2}$ in blocks of size $k$, starting at layer $1$ (the topmost block does not need to be a complete block). Let $l$ be the number of complete blocks in $B_1$. Let $a_1, a_2, \ldots$ be the topmost vertices in the topmost complete blocks in $B_1$ (possibly $a_1 = z_1$), that is, $a_1, a_2, \ldots$ are all vertices in $B_1$ in layer $kl$. Analoguously, let $b_1, b_2, \ldots$ be the topmost vertices in the topmost complete blocks in $B_2$ (possibly $b_1 = z_2$), that is, $b_1, b_2, \ldots$ are all vertices in $B_2$ in layer $kl$. Let $d(a_1, b_1)$ be the distance between $a_1$ and $b_1$ in $T$. We will distinguish two cases. 
	
	First, suppose that $1\leq d(a_1, b_1) \leq k$. 
	Consider the following partial coloring: 
	\begin{align*}
		c(v) = \begin{cases}
			i & \text{if $v\in B_1$ is in layer $i$, $i \equiv 0, -1 \pmod k$, $1\leq i\leq kl$}\\
			ki-1 & \text{if $v\in B_2$ is in layer $ki$, $1\leq ki \leq kl$}\\
			ki & \text{if $v\in B_2$ is in layer $ki-1$, $1\leq ki - 1 \leq kl$}. 
		\end{cases}
	\end{align*}
	Color the rest of the vertices with the highest color used in the partial coloring, that is, with color $kl$. 
	
	Notice that the partial coloring colors the top two vertices of every complete block, so the number of colors we used equals $2l$. 
	Let $u$ be a vertex in layer 0 in $B_1$ and $v$ a vertex in layer 0 in $B_2$. Notice that a shortest path from $u$ to $v$ uses at least one vertex from every complete block. It follows that $\diam(T^k) \geq 2l + 1$. So we conclude that the number of colors we used is at most $\diam(T^k)-1$. 
	
	We show that this coloring is a rainbow coloring. 
	Let $u \in B_1$ and $v\in B_2$ be two vertices in different sets of branches. Consider the path $P$ from $u$ to its ancestor $a_1$ using the layers that are $0 \pmod k$, and the path $Q$ from the ancestor $b_j$ of $v$ to $v$ using the layers that are $0 \pmod k$. Since $d(a_i, b_j) = d(a_1, b_1) \leq k$, there exists an edge $a_i b_j$ in $T^k$. Hence we can combine $P$ and $Q$ to a path from $u$ to $v$. This is a rainbow path. 
	
	Let $u$ and $v$ be two vertices in the same set of branches, so $u$ and $v$ are both in $B_1$ or both in $B_2$. Let $w$ be their lowest common ancestor, and let $i$ be the layer that contains $w$. If $i \leq kl$, then consider the path $P$ from $u$ to $w$ using the layer $0 \pmod k$ and the path $Q$ from $w$ to $v$ using the layers $-1 \pmod k$. Combining them yields a rainbow path from $u$ to $v$. 
	Suppose that $i > kl$. Notice that we cannot simply use the path $PQ$, since $w$ has the same color as the vertices in layer $kl$ or $kl-1$ (depending on whether $u$ and $v$ are in $B_1$ or $B_2$). Let $P$ and $Q$ be as before and let $p_j$ be the last vertex before $w$ in $P$ and $q_1$ the first vertex after $w$ in $Q$. Notice that $p_j$ is in layer $kl$ and $q_1$ is in layer $kl-1$. Then the distance between $p_j$ and $q_1$ in $T$ is at most $d(a_1, b_1) \leq k$, hence there is an edge $p_jq_1$ in $T^k$. So combining $P$ and $Q$ but excluding vertex $w$ is a path from $u$ to $v$, and it is a rainbow path.

	Second, suppose that $d(a_1, b_1) > k$. 
	Consider the following coloring:
	\begin{align*}
	c(v) = \begin{cases}
	i & \text{if $v\in B_1$ is in layer $i$, $i \equiv 0, -1 \pmod k$, $1\leq i \leq kl$}\\
	ki-1 & \text{if $v\in B_2$ is in layer $ki$, $1\leq ki \leq kl$}\\
	ki & \text{if $v\in B_2$ is in layer $ki-1$, $1\leq ki - 1 \leq kl$} \\
	1 & \text{otherwise}. 
	\end{cases}
	\end{align*}
	The number of colors we used in this coloring is $2l + 1$. The diameter of this graph is also $2l+2$: for a vertex $u$ in layer $0$ of $B_1$ and $v$ in layer $0$ of $B_2$, a shortest path needs at least one vertex from every complete block and a vertex from the topmost (incomplete) blocks. Such a path has length at least $2l+2$. So the number of colors we used is at most $\diam(T^k) - 1$. 
	
	We will show that this is a rainbow coloring. For vertices $u \in B_1$ and $v \in B_2$, we use the following path: from $u$ to $z_1$ using the layers $0 \pmod k$ and from $z_1$ to $v$ using the layers $0 \pmod k$. For two vertices $u$ and $v$ in the same set of branches, so both in $B_1$ or both in $B_2$, use the layers $0\pmod k$ to go from $u$ to the lowest common ancestor $w$ and the layers $-1 \pmod k$ to go from $w$ to $v$. Those paths are rainbow paths. 
\end{proof}
\begin{figure}
	\centering
	\begin{tikzpicture}
	\node[vertex] (z1) at (.75,0) {};
	\node[vertex] (z2) at (2.25,0) {};
	\node[vertex] (a1) at (0,-1) {};
	\node[vertex, label=left:$5$, label=right:$a_1$] (a2) at (0,-2) {};
	\node[vertex, label=left:$4$] (a3) at (0,-3) {};
	\node[vertex] (a4) at (0,-4) {};
	\node[vertex] (a5) at (0,-5) {};
	\node[vertex] (a6) at (0,-6) {};
	\node[vertex] (a7) at (0,-7) {};
	\node[vertex] (b1) at (1.5,-1) {};
	\node[vertex, label=left:$4$, label=right:$b_1$] (b2) at (1.5,-2) {};
	\node[vertex, label=left:$5$] (b3) at (1.5,-3) {};
	\node[vertex] (b4) at (1.5,-4) {};
	\node[vertex] (b5) at (1.5,-5) {};
	\node[vertex] (b6) at (1.5,-6) {};
	\node[vertex] (b7) at (1.5,-7) {};
	\node[vertex] (c1) at (3,-1) {};
	\node[vertex, label=left:$4$, label=right:$b_2$] (c2) at (3,-2) {};
	\node[vertex, label=left:$5$] (c3) at (3,-3) {};
	\node[vertex] (c4) at (3,-4) {};
	\node[vertex] (c5) at (3,-5) {};
	\node[vertex] (c6) at (3,-6) {};
	\node[vertex] (c7) at (3,-7) {};
	\draw[edge] (z1) -- (a1) -- (a2) -- (a3) -- (a4) -- (a5) -- (a6) -- (a7);
	\draw[edge] (z2) -- (b1) -- (b2) -- (b3) -- (b4) -- (b5) -- (b6) -- (b7);
	\draw[edge] (z2) -- (c1) -- (c2) -- (c3) -- (c4) -- (c5) -- (c6) -- (c7);
	\draw[edge] (z1) -- (z2);
	\end{tikzpicture}
	\qquad\qquad
	\begin{tikzpicture}
	\node[vertex, label=left:$10$, label=above:$a_1$] (z1) at (.75,3) {};
	\node[vertex, label=right:$9$, label=above:$b_1$] (z2) at (2.25,3) {};
	\node[vertex, label=left:$9$] (a-2) at (0,2) {};
	\node[vertex] (a-1) at (0,1) {};
	\node[vertex] (a0) at (0,0) {};
	\node[vertex] (a1) at (0,-1) {};
	\node[vertex, label=left:$5$] (a2) at (0,-2) {};
	\node[vertex, label=left:$4$] (a3) at (0,-3) {};
	\node[vertex] (a4) at (0,-4) {};
	\node[vertex] (a5) at (0,-5) {};
	\node[vertex] (a6) at (0,-6) {};
	\node[vertex] (a7) at (0,-7) {};
	\node[vertex, label=left:$10$] (b-2) at (1.5,2) {};
	\node[vertex] (b-1) at (1.5,1) {};
	\node[vertex] (b0) at (01.5,0) {};
	\node[vertex] (b1) at (1.5,-1) {};
	\node[vertex, label=left:$4$] (b2) at (1.5,-2) {};
	\node[vertex, label=left:$5$] (b3) at (1.5,-3) {};
	\node[vertex] (b4) at (1.5,-4) {};
	\node[vertex] (b5) at (1.5,-5) {};
	\node[vertex] (b6) at (1.5,-6) {};
	\node[vertex] (b7) at (1.5,-7) {};
	\node[vertex, label=left:$10$] (c-2) at (3,2) {};
	\node[vertex] (c-1) at (3,1) {};
	\node[vertex] (c0) at (3,0) {};
	\node[vertex] (c1) at (3,-1) {};
	\node[vertex, label=left:$4$] (c2) at (3,-2) {};
	\node[vertex, label=left:$5$] (c3) at (3,-3) {};
	\node[vertex] (c4) at (3,-4) {};
	\node[vertex] (c5) at (3,-5) {};
	\node[vertex] (c6) at (3,-6) {};
	\node[vertex] (c7) at (3,-7) {};
	\draw[edge] (z1) -- (a-2) -- (a-1) -- (a0) -- (a1) -- (a2) -- (a3) -- (a4) -- (a5) -- (a6) -- (a7);
	\draw[edge] (z2) -- (b-2) -- (b-1) -- (b0) -- (b1) -- (b2) -- (b3) -- (b4) -- (b5) -- (b6) -- (b7);
	\draw[edge] (z2) -- (c-2) -- (c-1) -- (c0) -- (c1) -- (c2) -- (c3) -- (c4) -- (c5) -- (c6) -- (c7);
	\draw[edge] (z1) -- (z2);
	\end{tikzpicture}
	\qquad\qquad
	\begin{tikzpicture}
	\node[vertex, label=left:$1$] (z1) at (.75,2) {};
	\node[vertex, label=right:$1$] (z2) at (2.25,2) {};
	\node[vertex] (a-1) at (0,1) {};
	\node[vertex] (a0) at (0,0) {};
	\node[vertex] (a1) at (0,-1) {};
	\node[vertex, label=left:$5$, label=right:$a_1$] (a2) at (0,-2) {};
	\node[vertex, label=left:$4$] (a3) at (0,-3) {};
	\node[vertex] (a4) at (0,-4) {};
	\node[vertex] (a5) at (0,-5) {};
	\node[vertex] (a6) at (0,-6) {};
	\node[vertex] (a7) at (0,-7) {};
	\node[vertex] (b-1) at (1.5,1) {};
	\node[vertex] (b0) at (01.5,0) {};
	\node[vertex] (b1) at (1.5,-1) {};
	\node[vertex, label=left:$4$, label=right:$b_1$] (b2) at (1.5,-2) {};
	\node[vertex, label=left:$5$] (b3) at (1.5,-3) {};
	\node[vertex] (b4) at (1.5,-4) {};
	\node[vertex] (b5) at (1.5,-5) {};
	\node[vertex] (b6) at (1.5,-6) {};
	\node[vertex] (b7) at (1.5,-7) {};
	\node[vertex] (c-1) at (3,1) {};
	\node[vertex] (c0) at (3,0) {};
	\node[vertex] (c1) at (3,-1) {};
	\node[vertex, label=left:$4$, label=right:$b_2$] (c2) at (3,-2) {};
	\node[vertex, label=left:$5$] (c3) at (3,-3) {};
	\node[vertex] (c4) at (3,-4) {};
	\node[vertex] (c5) at (3,-5) {};
	\node[vertex] (c6) at (3,-6) {};
	\node[vertex] (c7) at (3,-7) {};
	\draw[edge] (z1) -- (a-1) -- (a0) -- (a1) -- (a2) -- (a3) -- (a4) -- (a5) -- (a6) -- (a7);
	\draw[edge] (z2) -- (b-1) -- (b0) -- (b1) -- (b2) -- (b3) -- (b4) -- (b5) -- (b6) -- (b7);
	\draw[edge] (z2) -- (c-1) -- (c0) -- (c1) -- (c2) -- (c3) -- (c4) -- (c5) -- (c6) -- (c7);
	\draw[edge] (z1) -- (z2);
	\end{tikzpicture}
	\caption{Trees $T$, $k=5$. The numbers indicate the colouring of $T^k$. See Lemma \ref{lem:higher-center-two-vertices}. } \label{fig:higher-center-two-vertices}
\end{figure}
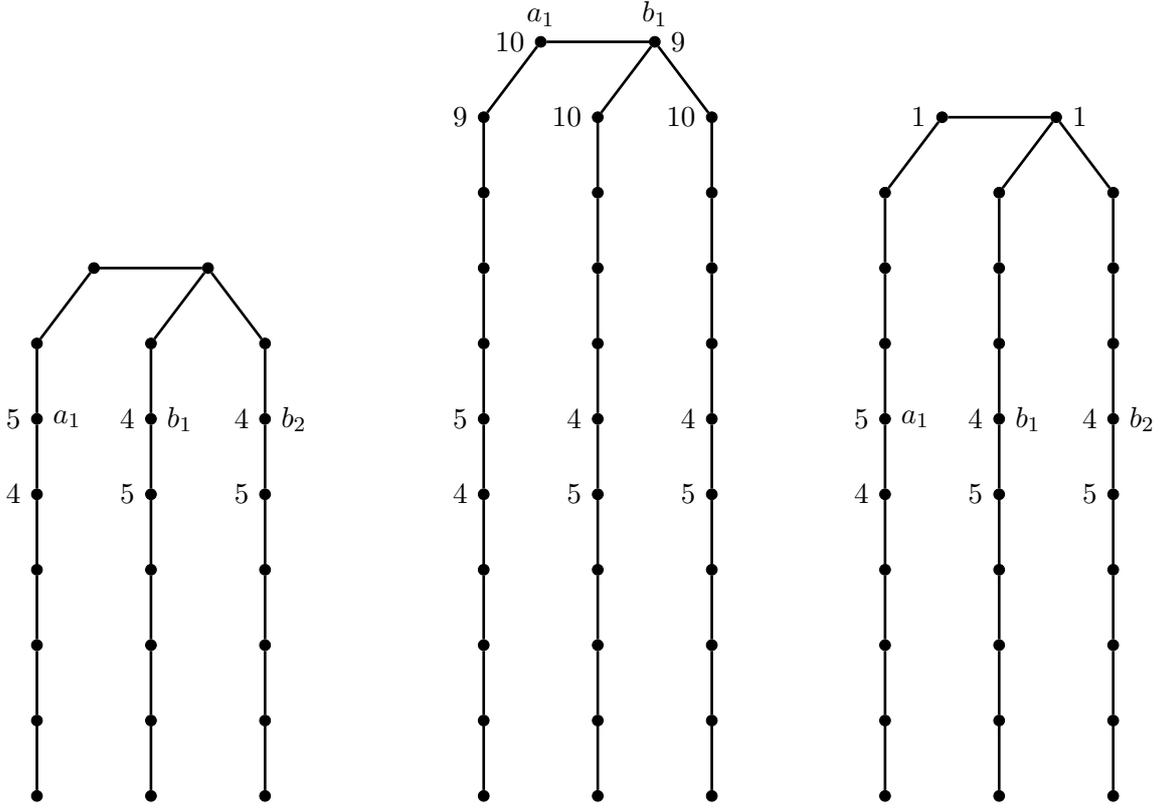

\begin{theorem}[=Theorem \ref{thm:powers-of-trees-Intro}] \label{thm:powers-of-trees}
	If $G$ is a power of a tree, then $\rvc(G)\in\{\diam(G)-1,\diam(G)\}$, and the corresponding optimal rainbow vertex coloring can be found in time that is linear in the size of $G$. 
\end{theorem}
\begin{proof}
	In Lemmas \ref{lem:square-upperbound}, \ref{lem:square-two-branches}, \ref{lem:square-center-two-vertices}, \ref{lem:higher-upperbound}, \ref{lem:higher-three-branches}, \ref{lem:higher-two-branches}, \ref{lem:higher-center-two-vertices}, it is shown that $\rvc(G)\in\{\diam(G)-1,\diam(G)\}$. 
	
	Suppose that $G = T^k$. If the tree $T$ is unknown, it can be computed in linear time \cite{Chang2015}. 
	First, we compute the center of $T$, and then we distinguish cases as in Lemmas \ref{lem:square-upperbound}, \ref{lem:square-two-branches}, \ref{lem:square-center-two-vertices}, \ref{lem:higher-upperbound}, \ref{lem:higher-three-branches}, \ref{lem:higher-two-branches}, \ref{lem:higher-center-two-vertices}. This costs linear time. In each of those lemmas, an optimal coloring is given that can be computed in linear time. 
\end{proof}

	\section{Conclusion and open problems}

In this work, we provided polynomial-time algorithms to rainbow vertex color permutation graphs, powers of trees, and split strongly chordal graphs. The algorithm provided for split strongly chordal graphs also works for the strong variant of the problem, where the rainbow paths connecting pairs of vertices are required to be also shortest paths. 

An interesting question to be answered towards solving Conjecture~\ref{conj:diametral} is whether \krvc~can be solved in polynomial time on AT-free graphs, i.e.\ graphs that do not contain an asteroidal triple. Conjecture~\ref{conj:diametral} has been proved true for interval graphs~\cite{MFCS2018} and, in this work, for permutation graphs, both of which are important subclasses of AT-free graphs.

Another direction of research within graph classes lies in determining the complexity of \krvc~and \ksrvc~on strongly chordal graphs. Note that both powers of trees and split strongly chordal graphs form subclasses of strongly chordal graphs for which RVC is polynomial-time solvable, as we show in this work. Finally, note that every strongly chordal graph is also a chordal graph, and the problems are known to be \NP-hard when restricted to chordal graphs.

	\bibliographystyle{siam}
	\bibliography{bib}

\appendix	
	\section{Shortest paths in permutation graphs} \label{app:shortest-paths}
	
	As mentioned in Section \ref{sec:permutation}, Lemma \ref{lem:X-Y-shortest} can also be found in \cite[Lemma 5]{MondalPP2003}, but for completeness we write the proof in this appendix. 
	We first prove some basic lemmas about paths in the intersection model. 
	
	\begin{lemma} \label{lem:intersect-intermediate-vertex}
		If $u \prec z \prec v$, then for every $u,v$-path $z_1(=u), z_2 \ldots, z_{a-1}, z_a (=v)$, there is a vertex $z_i$ that intersects $z$ (or equals $z$). 
	\end{lemma}
	\begin{proof}
		Suppose that none of the vertices $z_i$ intersects $z$. With induction we show that $z_i \prec z$ for all $1 \leq i \leq a$. By assumption $u = z_1 \prec z$. Suppose that $z_{i-1} \prec z$. Since $z_i \sim z_{i-1}$, we have that exactly one of $t(z_i) < t(z_{i-1}) < t(z)$ or $b(z_i) < b(z_{i-1}) < b(z)$. Because $z_i$ does not intersect $z$, it follows that $z_i \prec z$. This yields a contradiction with $z \prec v$. 
	\end{proof}
	
	\begin{lemma} \label{lem:shortest-path-intermediate-vertices}
		If $u \prec v$ and $z$ is a vertex in an induced $u,v$-path, then $z$ is not left of $u$ and not right of $v$. 
	\end{lemma}
	\begin{proof}
		Suppose that $z \prec u$. Then by Lemma \ref{lem:intersect-intermediate-vertex}, it follows that there is a vertex in the $z,v$-path that intersects $u$. This yields a contradiction with the fact that $u, \ldots, z, \ldots, v$ is an induced $u,v$-path.
		We conclude that $z$ is not left of $u$. Analogously, we see that $z$ is not right of $v$. 
	\end{proof}
	
	\begin{lemma} \label{lem:shortest-path-not-left}
		If $u \prec v$ and $z_1 (=u), z_2, z_3, \ldots, z_a (=v)$ is a shortest $u,v$-path, then for all $2< i < a -1$, it holds that $u \prec z_i \prec v$. 
	\end{lemma}
	\begin{proof}
		By Lemma \ref{lem:shortest-path-intermediate-vertices}, we know that $z_i$ is not left of $u$ for all $1 \leq i \leq a$. Since $z_1, z_2, z_3, \ldots, z_a$ is a shortest path, $z_i$ does not intersect $u$, for $2 < i$. Hence $u \prec z_i$ for all $i > 2$. Analogously, it holds that $z_i \prec v$ for $i< a-1$.  
	\end{proof}
		
	\begin{lemma} \label{lem:shortest-path-up-down}
		If $z_1, z_2, z_3, \ldots, z_a$ is a shortest $z_1,z_a$-path, then it either satisfies Equations \eqref{eq:x-even} and \eqref{eq:x-odd}, or Equations \eqref{eq:y-even} and \eqref{eq:y-odd}.
	\end{lemma}
	\begin{proof}
		We know that $z_2$ intersects $z_1$, so we have either $t(z_2) > t(z_{1})$ and $b(z_2) < b(z_{1})$, or $t(z_2) < t(z_{1})$ and $b(z_2) > b(z_{1})$. 
		
		Suppose that $t(z_2) > t(z_{1})$ and $b(z_2) < b(z_{1})$. Then we prove by induction that we are in the first case. Suppose that Equations~\ref{eq:x-even} and~\ref{eq:x-odd} hold for all $i<k$. Then consider $z_k$. Since $z_k$ is adjacent to $z_{k-1}$, we either have $t(z_k) > t(z_{k-1})$ and $b(z_k) < b(z_{k-1})$, or $t(z_k) < t(z_{k-1})$ and $b(z_k) > b(z_{k-1})$. 
		Suppose that $k$ is even (See Figure \ref{fig:shortest-path-up-down}). If $t(z_k) < t(z_{k-1})$ and $b(z_k) > b(z_{k-1})$, then, by the induction hypothesis, we know that $t(z_k) < t(z_{k-1})< t(z_{k-2})$ and $b(z_k) > b(z_{k-1}) > b(z_{k-2})$. Thus $z_k \sim z_{k-2}$. This yields a contradiction with the assumption that $z_1, z_2, z_3, \ldots, z_a$ is a shortest $z_1,z_a$-path. Hence, $t(z_k) > t(z_{k-1})$ and $b(z_k) < b(z_{k-1})$. 
		The case that $k$ is odd is analogous. 	
		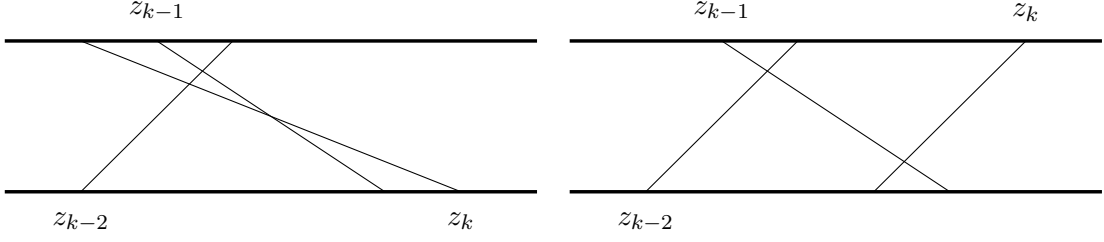
\begin{figure}
	\centering
	\begin{tikzpicture}
	\draw[line] (0,0) -- (7,0);
	\draw[line] (0,2) -- (7,2);
	
	\draw[segment] (1,0) -- (3,2);
	\draw[segment] (2,2) -- (5,0);
	\draw[segment] (1,2) -- (6,0);
	
	\node[label = below:$z_{k-2}$] (k2) at (1,0) {};
	\node[label = $z_{k-1}$] (k) at (2,2) {};
	\node[label = below:$z_{k}$] (k1) at (6,0) {};
	\end{tikzpicture}\quad
	\begin{tikzpicture}
		\draw[line] (0,0) -- (7,0);
		\draw[line] (0,2) -- (7,2);
		
		\draw[segment] (1,0) -- (3,2);
		\draw[segment] (2,2) -- (5,0);
		\draw[segment] (4,0) -- (6,2);
		
		\node[label = below:$z_{k-2}$] (k2) at (1,0) {};
		\node[label = $z_{k-1}$] (k) at (2,2) {};
		\node[label = $z_{k}$] (k1) at (6,2) {};
	\end{tikzpicture}
	\caption{See Lemma \ref{lem:shortest-path-up-down}. Since $z_k \sim z_{k-1}$, there are two cases: $t(z_k) < t(z_{k-1})$ or $t(z_k) > t(z_{k-1})$.} \label{fig:shortest-path-up-down}
\end{figure}
		
		Analogously, if $t(z_2) < t(z_{1})$ and $b(z_2) > b(z_{1})$, then we are in the second case. 
	\end{proof}
	
	\begin{lemma} \label{lem:X-short}
		Let $Z = z_1 = u, z_2, \ldots, z_a = v$ be a $u,v$-path and $t(z_2) > t(u)$. Then $X_{u,v}$ exists and the length of $Z$ is at least the length of $X_{u,v}$.
	\end{lemma}
	\begin{proof}
		Since $t(z_2) > t(u)$ and $z_2 \sim u$, there exists a vertex $x\sim u$ with $t(x) > t(u)$, hence the path $X_{u,v}$ exists. 

		Let $c-1$ be the length of $X_{u,v}$. 
		Suppose that $Z$ is a shorter path than $X_{u,v}$, that is, $a < c$. 
		In fact, let $Z$ be a shortest $u,v$-path with $t(z_2) > t(u)$.
		
		We will prove by induction that for $1 < i \leq a$, it holds that 
		\begin{align}
		&t(z_i) \leq t(x_{i}) \text{ and } b(z_i) \leq b(x_{i-1}) \text{ if $i$ is even,} \label{eq:z-even}\\ 
		&t(z_i) \leq t(x_{i-1}) \text{ and } b(z_i) \leq b(x_{i}) \text{ if $i$ is odd.}\label{eq:z-odd}
		\end{align}
		Intuitively, this means that $z_i$ is not right of $x_i$. See Figure \ref{fig:X-short} for the possible location of $z_{i}$ compared to $x_i$. 
		\begin{figure}
	\centering
	\begin{tikzpicture}
	\draw[line] (0,0) -- (7,0);
	\draw[line] (0,2) -- (7,2);
	
	\draw[segment] (.3,0) -- (2.5,2);
	\draw[segment] (1,2) -- (4,0);
	\draw[segment] (2,0) -- (6,2);
	\draw[segment] (4.5,2) -- (6.5,0);
	
	\node[label = $x_{i-2}$] (k) at (2.5,2) {};
	\node[label = below:$x_{i-1}$] (k1) at (4,0) {};
	\node[label = $x_{i}$] (k) at (6,2) {};
	\node[label = below:$x_{i+1}$] (v) at (6.5,0) {};
	
	\draw[thickline] (0,2) -- (2.5,2);
	\draw[thickline] (0,0) -- (4,0);
	\end{tikzpicture}\quad
	\begin{tikzpicture}
	\draw[line] (0,0) -- (7,0);
	\draw[line] (0,2) -- (7,2);
	
	\draw[segment] (.3,0) -- (2.5,2);
	\draw[segment] (1,2) -- (4,0);
	\draw[segment] (2,0) -- (6,2);
	\draw[segment] (4.5,2) -- (6.5,0);
	
	\node[label = below:$x_{i-1}$] (k1) at (4,0) {};
	\node[label = $x_{i}$] (k) at (6,2) {};
	\node[label = below:$x_{i+1}$] (v) at (6.5,0) {};

	\draw[thickline] (0,2) -- (6,2);
	\draw[thickline] (0,0) -- (4,0);
	\draw[possiblesegment] (.5,2) -- (3.5,0);
	\node[label = $z_{i-1}$] (zi) at (.5,2) {};
	\draw[possiblesegment] (1.25,0) -- (3.5,2);
	\node[label = $z_{i}$] (zi) at (3.5,2) {};
	\end{tikzpicture}
	\caption{The thick lines indicate where the segment of $z_{i-1}$ and $z_i$ can possibly end according to Equations \eqref{eq:z-even} and \eqref{eq:z-odd}. The dashed segments are examples of $z_{i-1}$ and $z_i$. See Lemma \ref{lem:X-short}.} \label{fig:X-short}
\end{figure}
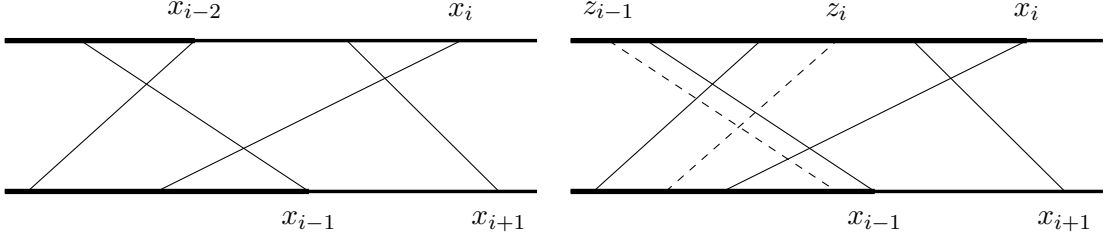
		
		We start with the base case $i=2$. We know that $t(z_2) \leq t(x_2)$ by the definition of $x_2$. And $b(z_2) < b(u) = b(x_1)$, since $z_2$ intersects $u$ and $t(z_2) > t(u)$. 
				
		Suppose that Equations \eqref{eq:z-even} and \eqref{eq:z-odd} hold for $i = k-1$, where $3 \leq k \leq a$. 
		Suppose that $k$ is even. 
		By Lemma \ref{lem:shortest-path-up-down}, we see that $b(z_k) < b(z_{k-1})$. By the induction hypothesis, we know that $b(z_{k-1}) \leq b(x_{k-1})$, thus $b(z_k) \leq b(x_{k-1})$. 
				
		Now we distinguish two cases: $z_k \sim x_{k-1}$ or $z_k \nsim x_{k-1}$. In the first case, by the definition of $x_k$, we have that $t(x_k) \geq t(z_k)$. In the second case, it holds that $t(z_k) \leq t(x_{k-1})$, since we already proved that $b(z_k) \leq b(x_{k-1})$. By Equation \eqref{eq:x-even}, we have that $t(x_k) > t(x_{k-1})$. We conclude that $t(z_k) \leq t(x_{k})$. 
				
		The case that $k$ is odd is analogous. 
				
		So, we conclude that 
		\begin{align*}
		&t(z_{a-1}) \leq t(x_{a-1}) \text{ and } b(z_{a-1}) \leq b(x_{a-2}) \text{ if $a-1$ is even,} \\ 
		&t(z_{a-1}) \leq t(x_{a-2}) \text{ and } b(z_{a-1}) \leq b(x_{a-1}) \text{ if $a - 1$ is odd.}
		\end{align*}
		Since we assumed that $a < c$ and $X_{u,v}$ is induced by definition, Lemma \ref{lem:shortest-path-intermediate-vertices} implies that $x_{a-1}$ and $x_{a-2}$ are both left of $v$. This yields a contradiction with the fact that $z_{a-1}$ intersects $v$. We conclude that $Z$ is at least as long as $X_{u,v}$. 
		\end{proof}
			
	\begin{lemma} \label{lem:Y-short}
		Let $Z = z_1 = u, z_2, \ldots, z_a = v$ be a $u,v$-path and $b(z_2) > b(u)$. Then the length of $Z$ is at least the length of $Y_{u,v}$.
	\end{lemma}
	\begin{proof}
		This is analogous to the proof of Lemma \ref{lem:X-short}. 
	\end{proof}
		
	\begin{proof}[Proof of Lemma \ref{lem:X-Y-shortest}]
		For every $u,v$-path $Z = z_1, z_2, \ldots, z_a$, it holds that either $t(z_2) > t(u)$ or $b(z_2) > b(u)$. 
		By \ref{lem:X-short} and \ref{lem:Y-short} it follows that the length of $Z$ is at least the minimum of  the  length of $X_{u,v}$ and the length of $Y_{u,v}$. Hence, at least one of $X_{u,v}$ and $Y_{u,v}$ is a shortest $u,v$-path. 
	\end{proof}

\end{document}